\documentclass[twocolumn,pra,aps,amssymb]{revtex4-1}
\usepackage[lined,algonl,boxed]{algorithm2e}
\usepackage{booktabs}
\usepackage{graphicx}
\usepackage{multirow}
\usepackage{amsmath,amssymb,amsfonts,amsthm}
\usepackage{multirow}
\usepackage{verbatim}
\usepackage{url}
\usepackage{comment}
\usepackage{mathtools}
\usepackage{color}
\usepackage{makecell}
\usepackage{dsfont}
\usepackage{hyperref}
\usepackage{lipsum}
\usepackage{capt-of}
\usepackage[most]{tcolorbox}
\newcommand{\ket}[1]{\ensuremath{\left|#1\right\rangle}}
\newcommand{\bra}[1]{\ensuremath{\left\langle#1\right|}}
\newtheorem{theorem}{Theorem}
\newtheorem{corollary}{Corollary}
\newtheorem{lemma}{Lemma}

\newtheorem{remark}{Remark}
\begin{document}

\title{Provable DI-QRNG protocols based on self-testing methodologies in preparation and measure scenario}
\author{Asmita Samanta$^1$\footnote{asmitasamanta6@gmail.com}, Arpita Maitra$^2$\footnote{arpita76b@gmail.com}, Goutam Paul$^3$\footnote{goutam.paul@isical.ac.in (Corresponding Author)}}
\affiliation{$^1$ Applied Statistics Unit, Indian Statistical Institute, Kolkata 700108, India.\\
$^2$ TCG Centres for Research and Education 
in Science and Technology (TCG CREST),
 Kolkata 700091, India.\\
$^3$ Cryptology and Security Research Unit, Indian Statistical Institute, Kolkata 700108, India.}

\begin{abstract}
We present two Device Independent Quantum Random Number Generator (DI-QRNG) protocols using two self-testing methodologies in Preparation \& Measure  (P\&M) scenario. These two methodologies are the variants of two well-known non-local games, namely, CHSH and pseudo-telepathy games, in P\&M framework. We exploit them as distinguishers in black-box settings to differentiate the classical and the quantum paradigms and hence to certify the Device Independence. The first self-test was proposed by Tavakoli et al. (Phys. Rev. A, 2018). We show that this is actually a P\&M variant of the CHSH game. Then based on this self-test, we design our first DI-QRNG protocol. We also propose a new self-testing methodology, which is the first of its kind that is reducible from pseudo-telepathy game in P\&M framework. Based on this new self-test, we design our second DI-QRNG protocol. 
\end{abstract}
\maketitle

\section{introduction}
The goal of self-testing is to validate whether a device is working as desired using only black-box type of interaction. In the quantum domain, self-testing is used as a mechanism to certify whether a protocol is device-independent, i.e., whether the protocol works correctly without any assumption on the devices, solely based on the measurement statistics. The question of device-independent security was first introduced in 1998 by Mayers and Yao in~\cite{mayers98} to take care of imperfect devices. Many years later in~\cite{barrett05}, it was shown that such security is indeed possible in principle, however, it uses a lot of classical communications which cost exponentially in the security parameter. This technique is well-established in Quantum Key Distribution (QKD) protocols~\cite{bb84,Ekert,six,semi,side2,side1} for certifying the security of the protocols independent to the underlying functionality of the devices~\cite{acin06a,acin06b,scarani06,acin07,pironio09,masanes09,hanggi10,masanes11,barrett12,reichardt12,vaz12} . 


In the similar fashion, Device-independent (in broader sense) Random Number Generation has been proposed in~\cite{redi1,redi2,redi3,redi4,fdi4,fdi5,si1,rp1,rp2}\footnote{Here, broader implies that some of the protocols are not fully device independent.} In~\cite{redi1}, it is shown that the violation of Bell’s inequality~\cite{Bell} or CHSH inequality~\cite{CHSH} could be used to certify randomness in a device-independent manner, i.e., without trusting the underlying devices. This was expanded theoretically by Colbeck and Kent~\cite{redi2}, who defined minimal assumptions for randomness certification in DI-QRNG and showed its potential for secure cryptographic applications. Liu et al. \cite{redi3} made significant experimental progress by achieving high generation rates using a loophole-free Bell test, marking a major advancement towards real-world applications. Brand\~{a}o et al. \cite{redi4} showed how quantum noise and decoherence impact randomness extraction, and then they offered protocols for managing these challenges in practical scenario. Coudron and Yuen \cite{fdi4} further contributed by integrating quantum theory and computational approaches, creating efficient randomness extraction protocols. Addressing partial randomness, Chung, Shi, and Wu \cite{fdi5} developed methods to amplify weak quantum randomness to near-perfect levels, making DI-QRNG more viable where instrumental imperfections might arise.

On the other hand, source independent random number generator was presented in~\cite{si1}. Semi Device-independent random number expansion is described in~\cite{rp1}, whereas Measurement Device-independent random number generation is reported in~\cite{rp2}.

 In~\cite{tavakoli}, Tavakoli et al. proposed a self-testing methodology in preparation and measure scenario without using any entanglement.
In their setting there are two black-boxes: one device prepares the quantum state $\rho$ and another measures the prepared state and generates a bit $b$. Both the devices accept bits as inputs. The preparation device outputs the quantum state depending on the input bits $x_0$ and $x_1$. The measurement device measures the prepared state based on the input bit $y$ and produces an output bit $b$. The two devices are connected by a quantum channel to transmit the quantum states from the preparation device to the measurement device. No classical communications are allowed between these two devices. A schematic diagram of the set-up is presented in Figure~\ref{fig}.
 \begin{figure}
\begin{center}
 \includegraphics[scale=0.4]{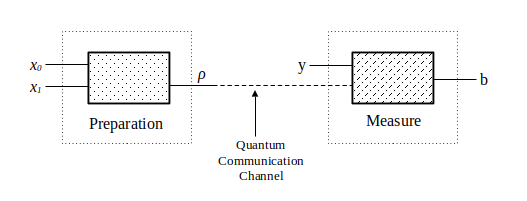}
 \caption{Schematic Diagram of the Self-testing Methodology~\cite{tavakoli}}
 \label{fig}
\end{center}
\end{figure}

In the above scheme, the two black boxes must be separated from each other so that no classical communication is allowed between the boxes during the process. Otherwise, it can be shown that the self-testing methodology can be simulated classically. Irrespective of whether the two devices are purchased from the same or two different vendors, there should only be a quantum channel as a connectivity between these two devices, that transmits a quantum state from the preparation device to the measurement one.

In this backdrop, we exploit the idea to propose a Device-independent Random Number Generation protocol in preparation and measure scenario. In this regard, one should mention the work~\cite{brunner14}, where a self-testing QRNG has been proposed in which the user can monitor the entropy in real time with some basic assumptions. However, the protocol is not fully device-independent. In the paper it is commented ``when these are
produced by trusted (or simply different) providers, it is
reasonable to assume that there are no (built-in) pre-established correlations between the devices". Contrary to this, in our protocol, we remove the assumption of ``trusted providers". We prove that under fair sampling and i.i.d. assumptions, simultaneous occurrence of Bell inequality violation, randomness extraction and deterministic measurements are not possible for any classically correlated devices in a non-signaling framework. Hence, there is no need to introduce ``trusted providers" in our case and the device independence security can be guaranteed.

We also present a new self-testing mechanism which can be reduced from Pseudo-Telepathy game and propose a novel DI-QRNG protocol.
 
\subsection*{Our Contributions}
The main contributions of our work are outlined below.
	
\begin{enumerate}
\item Tavakoli et al.~\cite{tavakoli} mentioned that their protocol does not ``necessarily relies on the violation of Bell's inequality"~\cite{tavakoli}. However, we prove that one can get an entangled version of the protocol (Theorem~\ref{theo}). We also show that the well-known CHSH game with quantum strategy can be reduced to this entangled version (Theorem~\ref{theorem}) and hence to the self-testing methodology presented in~\cite{tavakoli} (Corollary~\ref{cor}). 
\item Based on the above self-testing method, we propose our first DI-QRNG protocol. We prove the device independence of the protocol (Theorem~\ref{theo1} and~\ref{theo2}) and the true-randomness of its outputs (Theorem~\ref{theo5}).
\item Our next contribution is a novel self-testing game in P\&M scenario. We show that this game can be reduced from the pseudo-telepathy game~\cite{mptg} (Theorem~\ref{theo6}). All the existing self-testing methodologies so far including~\cite{tavakoli} directly rely on the violation of Bell's inequality or can be reduced to that. Contrary to this, ours is the first self-testing methodology in P\&M framework which can be reduced from the pseudo-telepathy game.
\item Finally, we propose a novel DI-QRNG protocol based on our new self-testing game and prove its device independence (Theorem~\ref{theo3}) and true-randomness of its outputs (Theorem~\ref{theo8}). 
\end{enumerate}


There are many companies that have launched their QRNG products in the market. These devices are usually designed in prepare-and-measure scenario which is comparatively easy to achieve in practice unlike entanglement-based products. However, those are not Device-independent in nature. In this direction, our proposed protocols may be implemented immediately to fulfil the need of commercial DI-QRNG.
	
\subsection*{Organization of the Paper}
The paper is organized as follows. In Section~\ref{prelim}, we recapitulate CHSH game (\ref{CHSH}),  multi-party pseudo-telepathy game (\ref{mpstg}) and the protocol presented in~\cite{tavakoli} (\ref{revi}). In Section~\ref{DI}, we redefine the quantum strategy of the original CHSH game keeping all the necessary parameters intact in the motivation towards connecting the game and the self-testing methodology of~\cite{tavakoli}. We name it as CHSH1 (\ref{CHSH1}).  In this regard, we present a new game $\mathcal{G}$ (\ref{entg}). Then we prove that the game $\mathcal{G}$ is the entanglement version of the self-testing methodology presented in~\cite{tavakoli} and show how it boils down to CHSH1, and hence to the CHSH game (\ref{g}). We then propose a DI-QRNG protocol based on the self-testing game in Section~\ref{proto} and give the proof of the device-independence and the proof of the true-randomness extraction of the proposed protocol in Section~\ref{pprf}. In Section~\ref{nstg}, we propose a new self-testing game in preparation and measure scenario, and show that it can be deduced from the multi-party pseudo-telepathy game. We then present a second DI-QRNG protocol based on the new game in Section~\ref{NQRNG}. In Section~\ref{prf}, we present the proof of device independence and the proof of the true randomness extraction of the proposed protocol. In Section~\ref{compare} we compare our two DI-QRNGs with source independent and measurement device independent QRNG protocols in P\&M scenario. 
Section~\ref{conl} concludes the paper.
\section{Preliminaries}
\label{prelim}
In this paper, either CHSH or pseudo-telepathy games are exploited as distinguishers to put a boundary between the classical-quantum paradigm. In case of CHSH game, the maximum winning probability in the classical case is $0.75$. A value greater than this implies some quantumness in the device and the upper bound of the probability in quantum domain is $0.85$. Hence, the difference between the maximum probabilities of the classical and the quantum paradigms is $0.1$. Interestingly, this difference increases further in case of the pseudo-telepathy game. Here, the the maximum classical winning probability is $\frac{1}{2} + 2^{-\lceil \frac{n}{2}\rceil}$, where $n$ is the number of players. However, in case of quantum strategy, the winning probability reaches $1$, providing a better distinguishability, up to a difference of 0.5 between the maximum probabilities of the classical and the quantum paradigms. 
\subsection{CHSH Game}
\label{CHSH}

Here we first recall the CHSH game \cite{CHSH}. Then we analyze its quantum strategy and show how it enhances the winning probability.

In the CHSH game, there are two participants (say Alice and Bob) and one dealer. The dealer sends a bit $x$ to Alice and a bit $y$ to Bob. After getting $x$ Alice outputs a bit $a$ and after getting $y$ Bob outputs a bit $b$. They win the game if $x \wedge y = a \oplus b$ holds.

\subsubsection{Classical Strategy}
As we know the truth table of AND function consists of 75\% 0s, the classical strategy to win the game is that Alice and Bob always return the same bit, say 0 (it also works for 1).\\

Using the classical strategy, the winning probability is 0.75 and this is optimal for the classical world. But instead of this strategy if they follow the quantum strategy then the winning probability can be enhanced to 0.85. Next, we discuss the quantum strategy.

\subsubsection{Quantum Strategy}

To use the quantum strategy, Alice and Bob have to share a maximally entangled state of the form $\frac{1}{\sqrt{2}}(\ket{00}+\ket{11}) $ among themselves beforehand. Their game strategy is as follows:\\
~\\
\textbf{Quantum Strategy of CHSH game:}
\begin{enumerate}
    \item After getting $x$ from the dealer, Alice measures her part of the entanglement.
    \begin{itemize}
        \item If $x$ is $0$, she measures her part in $\{\ket{0},\ket{1}\}$ basis.
        \item If $x$ is $1$, she measures her part in $\{\ket{+},\ket{-}\}$ basis.
    \end{itemize}
    
    \item After measurement if Alice gets $\ket{0}$ or $\ket{+}$, she returns $a=0$, but if Alice gets $\ket{1}$ or $\ket{-}$, she returns $a=1$.
    
    \item  After getting $y$ from the dealer, Bob measures his part.
    \begin{itemize}
        \item If $y$ is $0$, then he measures his part in $\{\ket{\psi},\ket{\psi^{\perp}}\}$ basis, where $\ket{\psi} = \cos \frac{\pi}{8} \ket{0} + \sin \frac{\pi}{8} \ket{1}$ and $\ket{\psi^{\perp}} = -\sin \frac{\pi}{8} \ket{0} + \cos \frac{\pi}{8} \ket{1}$ (please check $\langle\psi|\psi^{\perp}\rangle=0$).
        \item If $y$ is $1$, then he measures his part in $\{\ket{\phi}, \ket{\phi^{\perp}}\}$ basis, where $\ket{\phi}= \sin \frac{\pi}{8} \ket{0} + \cos \frac{\pi}{8} \ket{1}$ and $\ket{\phi^{\perp}} = -\cos \frac{\pi}{8} \ket{0} + \sin \frac{\pi}{8} \ket{1}$ (please check $\langle\phi|\phi^{\perp}\rangle=0$).
    \end{itemize}
    
    \item After measurement, if Bob gets $\ket{\psi}$ or $\ket{\phi^{\perp}}$, he returns $b=0$, but if Bob gets $\ket{\psi^{\perp}}$ or $\ket{\phi}$, he returns $b=1$.
    
    \item They win the game if $x \wedge y = a \oplus b$ holds.
\end{enumerate}
~\\
Table~\ref{tab1} calculates the probabilities of getting $(a,b)$ after receiving $(x,y)$. We use similar types of tables for all probability calculations in this work.

\begin{center}
\begin{small}
    \begin{tabular}{||c|c|c|c||}
        \hline
        $~(x,y)~$ & $~(a,b)~$ &  $~Pr[(a,b) | (x,y)]~$ & $~Pr[x \wedge y = a \oplus b | (x,y)]~$\\
        \hline \hline
        \multirow{4}{*}{$(0,0)$} & (0,0) & $\frac{1}{4}(1 + \frac{1}{\sqrt{2}})$ & $\frac{1}{4}(1 + \frac{1}{\sqrt{2}})$\\
        \cline{2-4}
         & $(0,1)$ & $\frac{1}{4}(1 - \frac{1}{\sqrt{2}})$ & $0$\\
         \cline{2-4}
         & $(1,0)$ & $\frac{1}{4}(1 - \frac{1}{\sqrt{2}})$ & $0$\\
         \cline{2-4}
         & $(1,1)$ & $\frac{1}{4}(1 + \frac{1}{\sqrt{2}})$ & $\frac{1}{4}(1 + \frac{1}{\sqrt{2}})$\\
        \hline
        \multirow{4}{*}{$(0,1)$} & (0,0) & $\frac{1}{4}(1 + \frac{1}{\sqrt{2}})$ & $\frac{1}{4}(1 + \frac{1}{\sqrt{2}})$\\
        \cline{2-4}
         & $(0,1)$ & $\frac{1}{4}(1 - \frac{1}{\sqrt{2}})$ & $0$\\
         \cline{2-4}
         & $(1,0)$ & $\frac{1}{4}(1 - \frac{1}{\sqrt{2}})$ & $0$\\
         \cline{2-4}
         & $(1,1)$ & $\frac{1}{4}(1 + \frac{1}{\sqrt{2}})$ & $\frac{1}{4}(1 + \frac{1}{\sqrt{2}})$\\
        \hline
        \multirow{4}{*}{$(1,0)$} & (0,0) & $\frac{1}{4}(1 + \frac{1}{\sqrt{2}})$ & $\frac{1}{4}(1 + \frac{1}{\sqrt{2}})$\\
        \cline{2-4}
         & $(0,1)$ & $\frac{1}{4}(1 - \frac{1}{\sqrt{2}})$ & $0$\\
         \cline{2-4}
         & $(1,0)$ & $\frac{1}{4}(1 - \frac{1}{\sqrt{2}})$ & $0$\\
         \cline{2-4}
         & $(1,1)$ & $\frac{1}{4}(1 + \frac{1}{\sqrt{2}})$ & $\frac{1}{4}(1 + \frac{1}{\sqrt{2}})$\\
        \hline
        \multirow{4}{*}{$(1,1)$} & (0,0) & $\frac{1}{4}(1 - \frac{1}{\sqrt{2}})$ & $0$\\
        \cline{2-4}
         & $(0,1)$ & $\frac{1}{4}(1 + \frac{1}{\sqrt{2}})$ & $\frac{1}{4}(1 + \frac{1}{\sqrt{2}})$\\
         \cline{2-4}
         & $(1,0)$ & $\frac{1}{4}(1 + \frac{1}{\sqrt{2}})$ & $\frac{1}{4}(1 + \frac{1}{\sqrt{2}})$\\
         \cline{2-4}
         & $(1,1)$ & $\frac{1}{4}(1 - \frac{1}{\sqrt{2}})$ & $0$\\
        \hline
    \end{tabular}
    \captionof{table}{Winning probability of the CHSH game}
    \label{tab1}
\end{small}
\end{center}

So, in each case, the winning probability is $$\left(2 \cdot \frac{1}{4} \left(1 + \frac{1}{\sqrt{2}}\right)\right) = \frac{1}{2} \left(1 + \frac{1}{\sqrt{2}}\right) \approx 0.85.$$
\subsection{Multi-Party Pseudo-Telepathy Game}\label{mpstg}

Pseudo-Telepathy game can be played in various ways, even with two players~\cite{ptg2}. However, in the current initiative, we are considering the Greenberger–Horne–Zeilinger (GHZ) game which can not be played less than 3 players~\cite{mptg}.

Let us name all the players taking part in this game as $A_i$ for $i\in \{1,2,...,n\}$. After starting the game the dealer sends a bit $x_i$ to the player $A_i$ and at the end of the game the player $A_i$ has to return a bit $y_i$, for $i\in \{1,2,...,n\}$. Here the condition on the input bits $x_i$ is that \[\sum_{i=1}^{n} x_i = 0 ~\text{(mod 2).}\] The players win the game if and only if (in short, iff) the following condition is satisfied. \[\sum_{i=1}^{n} y_i = \frac{1}{2}\sum_{i=1}^{n} x_i ~\text{(mod 2)}\]
That means the winning condition is basically the following:
\[\sum_{i=1}^{n} y_i =\begin{cases}
\text{0}, & \text{if $\sum_{i=1}^{n} x_i = 0$ (mod 4)}\\
\text{1}, & \text{if $\sum_{i=1}^{n} x_i = 2$ (mod 4)}
\end{cases}
\]
No communication is allowed among the $n$ participants after receiving the inputs and before producing the outputs.
In the paper \cite{mptg}, it is shown that the winning probability of this game with classical strategies is at most $\frac{1}{2} + 2^{-\lceil \frac{n}{2}\rceil}$, where $n$ is the number of the players. However, the winning probability of this game with quantum strategies is $1$. For detail strategy one may consult the paper \cite{mptg}. Here, we briefly describe the strategy and show how this works.

 Define 
$$|\Phi_n^{+}\rangle = 
\frac{1}{\sqrt{2}} |0^n\rangle + \frac{1}{\sqrt{2}} |1^n\rangle$$ and
$$|\Phi_n^{-}\rangle = 
\frac{1}{\sqrt{2}} |0^n\rangle - \frac{1}{\sqrt{2}} |1^n\rangle.$$
$H$ denotes Hadamard transform. $S$ denotes the unitary transformation
$S|0\rangle \mapsto |0\rangle, \ S|1\rangle \mapsto i|1\rangle$.
If $S$ is applied to any two qubits of $|\Phi_n^{+}\rangle$ leaving the 
other qubits undisturbed then the resulting state is 
$|\Phi_n^{-}\rangle$ and vice versa.
If $|\Phi_n^{+}\rangle$ is distributed among $n$ players and if exactly
$m$ of them apply $S$ to their qubit, then the resulting global state will be 
$|\Phi_n^{+}\rangle$ if $m \equiv 0 \bmod 4$ and 
$|\Phi_n^{-}\rangle$ if $m \equiv 2 \bmod 4$.
Note that $$(H^{\otimes n})|\Phi_n^{+}\rangle
= \frac{1}{\sqrt{2^{n-1}}} \sum_{wt(y) \equiv 0 \bmod 2}|y\rangle$$
and $$(H^{\otimes n})|\Phi_n^{-}\rangle
= \frac{1}{\sqrt{2^{n-1}}} \sum_{wt(y) \equiv 1 \bmod 2}|y\rangle.$$
The players are allowed to share a prior entanglement, $|\Phi_n^{+}\rangle$. 
\begin{enumerate}
\item If $x_i = 1$, $A_i$ applies transformation $S$ to his qubit; otherwise 
he does nothing.
\item He applies $H$ to his qubit.
\item He measures his qubit in order to obtain $y$.
\item He produces $y_i$ as his output.
\end{enumerate}
Hence, the game $G_n$ is always won by the $n$ distributed parties without 
any communication among themselves.

So as the winning probability differences in two cases is almost $0.5$ for large number of $n$, it is easy to distinguish whether the classical strategy has been used or the quantum strategy has been used from the winning probability. In our paper, we denote the multi-party pseudo-telepathy game for three parties as $G_1$ and use it to design our new self-testing game.

\subsection{Revisiting the protocol by Tavakoli et al.~\cite{tavakoli}}
\label{revi}
In the paper \cite{tavakoli}, the authors have proposed a {\it preparation and measure} based protocol to test the device independence assuming an upper bound on the Hilbert space dimension. The protocol is as follows:\\

\begin{itemize}
    \item Consider a two-bit input $x = x_0 x_1$ and a one-bit input $y$, where $x_0, x_1, y \in \{0,1\}$.
    \item After receiving the bits $x_0, x_1,$ the protocol prepares some qubits and depending on $y$ it measures the qubits and returns a bit $b$.
    \item The algorithm is successful if the returned bit $b$ is the same as the input bit $x_y$.\\
\end{itemize}

In other words, we can say that the success probability of the algorithm is $$\mathcal{A} = \frac{1}{8} \sum_{x_0,x_1,y} Pr[b=x_y | x_0,x_1,y].$$

To achieve the optimality of the success probability, the authors have suggested the following strategy:\\
~\\
\textbf{Optimal Success Strategy for the protocol by Tavakoli et al.:}

\begin{enumerate}
    \item After receiving the bits $x_0, x_1$, the protocol prepares a qubit as follows:
    \begin{itemize}
        \item If $x$ is $00$ then the pure state is $\rho_{00} = \frac{\mathds{1} + \sigma_x}{2}.$ This is basically the qubit $\ket{+}$ as $\rho_{00}=\ket{+}\bra{+}$.
        \item If $x$ is $01$ then the pure state is $\rho_{01} = \frac{\mathds{1} + \sigma_z}{2}.$ This is basically the qubit $\ket{0}$ as $\rho_{01}=\ket{0}\bra{0}$.
        \item If $x$ is $10$ then the pure state is $\rho_{10} = \frac{\mathds{1} - \sigma_z}{2}.$ This is basically the qubit $\ket{1}$ as $\rho_{10}=\ket{1}\bra{1}$.
        \item If $x$ is $11$ then the pure state is $\rho_{11} = \frac{\mathds{1} - \sigma_x}{2}.$ This is basically the qubit $\ket{-}$ as $\rho_{11}=\ket{-}\bra{-}$.
    \end{itemize}
    
    \item Then depending on $y$, it chooses a measurement basis as follows:
    \begin{itemize}
        \item If $y$ is $0$, then it chooses $M_0 = \frac{\sigma_x + \sigma_z}{\sqrt{2}}$. This is actually a projection on eigenspace of the vector $\vec{v}.\vec{\sigma}$, where $\vec{v}=\frac{1}{\sqrt{2}}(1,0,1)$ and $\sigma_1=\sigma_x$, $\sigma_2=\sigma_y$ and $\sigma_3=\sigma_z$~\cite{nelson}. In other words, this is the measurement of qubits in $\{ \ket{\psi}, \ket{\psi^{\perp}}\}$ basis, where $\ket{\psi} = \cos \frac{\pi}{8} \ket{0} + \sin \frac{\pi}{8} \ket{1}$ and $\ket{\psi^{\perp}} = -\sin \frac{\pi}{8} \ket{0} + \cos \frac{\pi}{8} \ket{1}$.
        \item If $y$ is $1$, then it chooses $M_1 = \frac{\sigma_x - \sigma_z}{\sqrt{2}}$. This is a projection on eigenspace of the vector $\vec{v}.\vec{\sigma}$, where $\vec{v}=\frac{1}{\sqrt{2}}(1,0,-1)$. On other words, we are actually measuring the qubit in $\{ \ket{\phi}, \ket{\phi^{\perp}}\}$ basis, where $\ket{\phi} = \sin \frac{\pi}{8} \ket{0} + \cos \frac{\pi}{8} \ket{1}$ and $\ket{\phi^{\perp}} = -\cos \frac{\pi}{8} \ket{0} + \sin \frac{\pi}{8} \ket{1}$.
    \end{itemize}
    
    \item The returned bit $b$ is $0$, if after measurement it gets $\ket{\psi}$ or $\ket{\phi}$. The bit $b$ is $1$, if after measurement it gets $\ket{\psi^{\perp}}$ or $\ket{\phi^{\perp}}$. However, to calculate the value of $\mathcal{A}$, we  only consider the instances where $b=x_y$, i.e., $b$ takes the value of the $y$th bit of $x$.
\end{enumerate}
~\\
In Table~\ref{tavakoli}, we observe the probabilities of getting $b(=x_y)$ to calculate the value of $\mathcal{A}$.\\

\begin{center}
    \begin{tabular}{||c|c|c|c|c||}
        \hline
        $~x (x_0x_1)~$  & $~~\rho_x~~$ & $~y~$ & $~x_y~$ & $~Pr[b=x_y | x_0,x_1,y]~$\\
        \hline \hline
        \multirow{2}{*}{$00$} & \multirow{2}{*}{$\ket{+}$} & $0$ & $0$ & $\frac{1}{2}(1 + \frac{1}{\sqrt{2}})$\\
        \cline{3-5}
         &  & $1$ & $0$ & $\frac{1}{2}(1 + \frac{1}{\sqrt{2}})$\\
        \hline
        \multirow{2}{*}{$01$} & \multirow{2}{*}{$\ket{0}$} & $0$ & $0$ & $\frac{1}{2}(1 + \frac{1}{\sqrt{2}})$\\
        \cline{3-5}
         &  & $1$ & $1$ & $\frac{1}{2}(1 + \frac{1}{\sqrt{2}})$\\
        \hline
        \multirow{2}{*}{$10$} & \multirow{2}{*}{$\ket{1}$} & $0$ & $1$ & $\frac{1}{2}(1 + \frac{1}{\sqrt{2}})$\\
        \cline{3-5}
         &  & $1$ & $0$ & $\frac{1}{2}(1 + \frac{1}{\sqrt{2}})$\\
        \hline
        \multirow{2}{*}{$11$} & \multirow{2}{*}{$\ket{-}$} & $0$ & $1$ & $\frac{1}{2}(1 + \frac{1}{\sqrt{2}})$\\
        \cline{3-5}
         &  & $1$ & $1$ & $\frac{1}{2}(1 + \frac{1}{\sqrt{2}})$\\
        \hline
    \end{tabular}
    \captionof{table}{Winning probability of the protocol~\cite{tavakoli}}
    \label{tavakoli}
\end{center}
Now,
\begin{center}
    $\mathcal{A} = \frac{1}{8} (8 \times (\frac{1}{2}(1 + \frac{1}{\sqrt{2}})))$ $= \frac{1}{2}(1 + \frac{1}{\sqrt{2}})$ $\approx 0.85$.
\end{center}

In our protocol $\mathcal{P}$ towards {\it Device-Independent QRNG}, we have used the above optimal methodology.

\section{Equivalence of CHSH game to the Self-Testing Protocol by Tavakoli et al.~\cite{tavakoli}} 
\label{DI}
One of the  key result of our paper is that the CHSH game is the same as the self-testing game, but in
the CHSH game, an entangled state is required, whereas in the self-testing game,
there is no entanglement, and the preparation and measurement are independent. 
This is an intrinsic difference that we highlight and explain why they
are the same in this section.


We have already discussed the CHSH game (\ref{CHSH}) and its quantum strategy. Here, we redefine the strategy slightly and show how it merges to the entangled version of the self-testing protocol by Tavakoli et al. Here, we exploit the unique property of the entanglement that if a subsystem is measured, the entanglement is broken and the counter subsystem will collapse to a quantum state depending on the measurement basis chosen for the former subsystem (Theorem 1).

\subsection{Defining CHSH1 from CHSH Game}
\label{CHSH1}
We have already seen the quantum strategy of the CHSH game. In the quantum strategy of the CHSH game whenever we get $x=0$ we use the $\{\ket{0},\ket{1}\}$ basis and we have used the $\{\ket{+},\ket{-}\}$ basis for $x=1$. But in the protocol of Tavakoli et al., for preparation of quantum states they have used $\{\ket{+},\ket{-}\}$ basis whenever $(x_0 \oplus x_1) = 0$ and $\{\ket{0},\ket{1}\}$ basis whenever $(x_0 \oplus x_1) = 1$. So, here we redefine the previous quantum strategy of CHSH game a little bit just for a easy reduction process. All the modifications are presented in italics font. Note that such changes do not affect the winning condition of the game. The modifications make it easy to establish the one-to-one mapping between the CHSH game and the methodology presented by Tavakoli et al.~\cite{tavakoli}. The redefined game strategy is as follows:\\
~\\
\textbf{Redefined Strategy for CHSH1 Game:}

\begin{enumerate}
    \item After getting $x$ from the dealer, Alice measures her part of the entanglement as follows.
    \begin{itemize}
        \item  {\em If $x$ is $0$, she measures her part in $\{\ket{+},\ket{-}\}$ basis.}
        \item {\em If $x$ is $1$, she measures her part in $\{\ket{0},\ket{1}\}$ basis.}
    \end{itemize}
    
    \item After measurement if Alice gets $\ket{0}$ or $\ket{+}$, she returns $a=0$, but if Alice gets $\ket{1}$ or $\ket{-}$, she returns $a=1$.
     
    \item Next Bob measures his part.
    \begin{itemize}
        \item If $y$ is $0$, then he measures his part in $\{\ket{\psi}, \ket{\psi^{\perp}}\}$ basis, where $\ket{\psi} = \cos \frac{\pi}{8} \ket{0} + \sin \frac{\pi}{8} \ket{1}$ and $\ket{\psi^{\perp}} = -\sin \frac{\pi}{8} \ket{0} + \cos \frac{\pi}{8} \ket{1}$ (as described earlier).
        \item If $y$ is $1$, then he measures his part in $\{\ket{\phi}, \ket{\phi^{\perp}}\}$ basis, where $\ket{\phi} = \sin \frac{\pi}{8} \ket{0} + \cos \frac{\pi}{8} \ket{1}$ and $\ket{\phi^{\perp}} = -\cos \frac{\pi}{8} \ket{0} + \sin \frac{\pi}{8} \ket{1}$ (as described earlier).
    \end{itemize}
     
    \item {\em After measurement, if Bob gets $\ket{\psi}$ or $\ket{\phi}$ he returns $b=0$, but if Bob gets $\ket{\psi^{\perp}}$ or $\ket{\phi^{\perp}}$ he returns $b=1$.}
     
    \item They win the game if $x \wedge y = a \oplus b$ holds.
\end{enumerate}
~\\
Here, in Table~\ref{redifine}, we calculate the winning probability with this redefined strategy.\\

\begin{center}
\begin{small}
    \begin{tabular}{||c|c|c|c||}
        \hline
        $~(x,y)~$ & $~(a,b)~$ &  $~Pr[(a,b) | (x,y)]~$ & $~Pr[x \wedge y = a \oplus b | (x,y)]~$\\
        \hline \hline
        \multirow{4}{*}{$(0,0)$} & (0,0) & $\frac{1}{4}(1 + \frac{1}{\sqrt{2}})$ & $\frac{1}{4}(1 + \frac{1}{\sqrt{2}})$\\
        \cline{2-4}
         & $(0,1)$ & $\frac{1}{4}(1 - \frac{1}{\sqrt{2}})$ & $0$\\
         \cline{2-4}
         & $(1,0)$ & $\frac{1}{4}(1 - \frac{1}{\sqrt{2}})$ & $0$\\
         \cline{2-4}
         & $(1,1)$ & $\frac{1}{4}(1 + \frac{1}{\sqrt{2}})$ & $\frac{1}{4}(1 + \frac{1}{\sqrt{2}})$\\
        \hline
        \multirow{4}{*}{$(0,1)$} & (0,0) & $\frac{1}{4}(1 + \frac{1}{\sqrt{2}})$ & $\frac{1}{4}(1 + \frac{1}{\sqrt{2}})$\\
        \cline{2-4}
         & $(0,1)$ & $\frac{1}{4}(1 - \frac{1}{\sqrt{2}})$ & $0$\\
         \cline{2-4}
         & $(1,0)$ & $\frac{1}{4}(1 - \frac{1}{\sqrt{2}})$ & $0$\\
         \cline{2-4}
         & $(1,1)$ & $\frac{1}{4}(1 + \frac{1}{\sqrt{2}})$ & $\frac{1}{4}(1 + \frac{1}{\sqrt{2}})$\\
        \hline
        \multirow{4}{*}{$(1,0)$} & (0,0) & $\frac{1}{4}(1 + \frac{1}{\sqrt{2}})$ & $\frac{1}{4}(1 + \frac{1}{\sqrt{2}})$\\
        \cline{2-4}
         & $(0,1)$ & $\frac{1}{4}(1 - \frac{1}{\sqrt{2}})$ & $0$\\
         \cline{2-4}
         & $(1,0)$ & $\frac{1}{4}(1 - \frac{1}{\sqrt{2}})$ & $0$\\
         \cline{2-4}
         & $(1,1)$ & $\frac{1}{4}(1 + \frac{1}{\sqrt{2}})$ & $\frac{1}{4}(1 + \frac{1}{\sqrt{2}})$\\
        \hline
        \multirow{4}{*}{$(1,1)$} & (0,0) & $\frac{1}{4}(1 - \frac{1}{\sqrt{2}})$ & $0$\\
        \cline{2-4}
         & $(0,1)$ & $\frac{1}{4}(1 + \frac{1}{\sqrt{2}})$ & $\frac{1}{4}(1 + \frac{1}{\sqrt{2}})$\\
         \cline{2-4}
         & $(1,0)$ & $\frac{1}{4}(1 + \frac{1}{\sqrt{2}})$ & $\frac{1}{4}(1 + \frac{1}{\sqrt{2}})$\\
         \cline{2-4}
         & $(1,1)$ & $\frac{1}{4}(1 - \frac{1}{\sqrt{2}})$ & $0$\\
        \hline
    \end{tabular}
   \end{small}
    \captionof{table}{Winning probability with redefined strategy}
    \label{redifine}
\end{center}

Here the table is the same as the previous Table~\ref{tab1}. So, in each case, the winning probability is $$\left(2 \cdot \frac{1}{4} \left(1 + \frac{1}{\sqrt{2}}\right)\right) = \frac{1}{2} \left(1 + \frac{1}{\sqrt{2}}\right) \approx 0.85.$$\\
Here basically the game (i.e., players win if $x \wedge y = a \oplus b$ holds) is same as CHSH game but as the quantum strategy is a little bit different, we call it  CHSH1. 

\subsection{An entangled version $\mathcal{G}$ of the self-testing protocol by Tavakoli et al.~\cite{tavakoli}}
\label{entg}
 In this section, we introduce a new game $\mathcal{G}$, defined as follows.
\begin{itemize}
 \item Let there be two participants (say, Alice and Bob) and one dealer. 
 \item Dealer sends a two-bit input $x = x_0 x_1$ to Alice and a one-bit input $y$ to Bob. 
 \item After getting $x = x_0 x_1$, Alice outputs a bit $a$ and after getting $y$, Bob outputs a bit $b$. 
 \item They win the game,  if $(x_0 \oplus x_1) \wedge y = a \oplus b$ holds.
\end{itemize}
~\\
\textbf{Optimal winning probability of the game $\mathcal{G}$ using Classical Strategy:~}\\

The maximum winning probability can not exceed 0.75 using any classical strategy. In fact, if we can find a classical strategy (say, Strategy 1) for this new game to win with a probability greater than 0.75 then we can design a classical strategy, using that strategy (Strategy 1), to win CHSH1 game with more than 0.75 probability as follows:\\
\begin{itemize}
    \item Get one-bit inputs $x$ and $y$.
    \item Randomly choose a bit $r \in \{0,1\}$, and set $x_0 = r$ and $x_1 = r \oplus x$.
    \item Use Strategy 1 to get outputs $a, b$ for inputs $(x_0 x_1), y$ such that $(x_0 \oplus x_1) \wedge y = a \oplus b$ holds.
    \item Since here $(x_0 \oplus x_1) \wedge y = a \oplus b \implies x \wedge y = a \oplus b$, just outputs $a,b$ as of Strategy 1.\\
\end{itemize}

If the winning probability of Strategy 1 is more than 0.75, then the winning probability of the classical strategy for CHSH1 game is also more than 0.75, which is a contradiction as CHSH1 game is equivalent with CHSH game. So, the maximum winning probability for the new game can not exceed 0.75 using any classical strategy.\\
~\\
\textbf{Quantum Strategy of the game $\mathcal{G}$:~}\\

In quantum strategy, like CHSH1 game, Alice and Bob share a maximally entangled state of the form $\frac{1}{\sqrt{2}}(\ket{00}+\ket{11})$. The optimal strategy for the game $\mathcal{G}$ is as follows:\\
~\\
\textbf{Optimal Quantum Strategy for Game $\mathcal{G}$:}

\begin{enumerate}
    \item After getting $(x_0,x_1)$ from the dealer, Alice computes $(x_0 \oplus x_1)$ and measure her part of the entanglement.
    \begin{itemize}
        \item  If $(x_0 \oplus x_1)$ is $0$, she measures her part in $\{\ket{+},\ket{-}\}$ basis.
        \item If $(x_0 \oplus x_1)$ is $1$, she measures her part in $\{\ket{0},\ket{1}\}$ basis.
    \end{itemize}
    
    \item After measurement if Alice gets $\ket{0}$ or $\ket{+}$, she returns $a=0$, but if Alice gets $\ket{1}$ or $\ket{-}$, she returns $a=1$.
    
    \item Next Bob measures his part.
    \begin{itemize}
        \item If $y$ is $0$, then he measures his part in $\{\ket{\psi}, \ket{\psi^{\perp}}\}$ basis, where $\ket{\psi} = \cos \frac{\pi}{8} \ket{0} + \sin \frac{\pi}{8} \ket{1}$ and $\ket{\psi^{\perp}} = -\sin \frac{\pi}{8} \ket{0} + \cos \frac{\pi}{8} \ket{1}$.
        \item If $y$ is $1$, then he measures his part in $\{\ket{\phi}, \ket{\phi^{\perp}}\}$ basis, where $\ket{\phi} = \sin \frac{\pi}{8} \ket{0} + \cos \frac{\pi}{8} \ket{1}$ and $\ket{\phi^{\perp}} = -\cos \frac{\pi}{8} \ket{0} + \sin \frac{\pi}{8} \ket{1}$.
    \end{itemize}
    
    \item After measurement, if Bob gets $\ket{\psi}$ or $\ket{\phi}$ he returns $b=0$, but if Bob gets $\ket{\psi^{\perp}}$ or $\ket{\phi^{\perp}}$ he returns $b=1$.
    
    \item They win the game if $(x_0 \oplus x_1) \wedge y = a \oplus b$ holds.
\end{enumerate}
~\\
The probabilities of winning the game $\mathcal{G}$ in various instances for the quantum strategy are shown in Table~\ref{table:1}.

\begin{widetext}
    \begin{minipage}{\linewidth}
\begin{center}
    \begin{tabular}{||c|c|c|c|c|c|c||}
        \hline
        &&~Alice's state~&~Bob's state~&~Bob's state~&&\\
        $~(x_0,x_1,y)~$&$~x_0\oplus x_1~$&~after~&~before~&~after~&$~(a,b)~$&$~\Pr[(x_0 \oplus x_1) \wedge y = a \oplus b | (x,y)]~$\\
        &&~measurement~&~measurement~&~measurement~&&\\
        \hline \hline
         &0&\ket{+}&\ket{+}&\ket{\psi}& $(0,0)$& $\frac{1}{4}(1 + \frac{1}{\sqrt{2}})$\\
        \cline{2-7}
        $(0,0,0)$ or&0&\ket{+}&\ket{+}&\ket{\psi^{\perp}}& $(0,1)$ & $0$\\
         \cline{2-7}
        $(1,1,0)~~~~$ &0&\ket{-}&\ket{-}&\ket{\psi}& $(1,0)$ & $0$\\
         \cline{2-7}
         &0&\ket{-}&\ket{-}&\ket{\psi^{\perp}}& $(1,1)$ & $\frac{1}{4}(1 + \frac{1}{\sqrt{2}})$\\
        \hline\hline
         &0&\ket{+}&\ket{+}&\ket{\phi}& $(0,0)$ & $\frac{1}{4}(1 + \frac{1}{\sqrt{2}})$\\
        \cline{2-7}
        $(0,0,1)$ or &0&\ket{+}&\ket{+}&\ket{\phi^{\perp}}& $(0,1)$ & $0$\\
         \cline{2-7}
        $(1,1,1)~~~~$ &0&\ket{-}&\ket{-}&\ket{\phi}& $(1,0)$ & $0$\\
         \cline{2-7}
         &0&\ket{-}&\ket{-}&\ket{\phi^{\perp}}& $(1,1)$ & $\frac{1}{4}(1 + \frac{1}{\sqrt{2}})$\\
         \hline\hline
         &1&\ket{0}&\ket{0}&\ket{\psi}& $(0,0)$ & $\frac{1}{4}(1 + \frac{1}{\sqrt{2}})$\\
        \cline{2-7}
        $(0,1,0)$ or &1&\ket{0}&\ket{0}&\ket{\psi^{\perp}}& $(0,1)$ & $0$\\
         \cline{2-7}
        $(1,0,0)~~~~$ &1&\ket{1}&\ket{1}&\ket{\psi}& $(1,0)$ & $0$\\
         \cline{2-7}
         &1&\ket{1}&\ket{1}&\ket{\psi^{\perp}}& $(1,1)$ & $\frac{1}{4}(1 + \frac{1}{\sqrt{2}})$\\
         \hline\hline
         &1&\ket{0}&\ket{0}&\ket{\phi}& $(0,0)$ & $0$\\
        \cline{2-7}
        $(0,1,1)$ or &1&\ket{0}&\ket{0}&\ket{\phi^{\perp}}& $(0,1)$ & $\frac{1}{4}(1 + \frac{1}{\sqrt{2}})$\\
         \cline{2-7}
        $(1,0,1)~~~~$ &1&\ket{1}&\ket{1}&\ket{\phi}& $(1,0)$ & $\frac{1}{4}(1 + \frac{1}{\sqrt{2}})$\\
         \cline{2-7}
         &1&\ket{1}&\ket{1}&\ket{\phi^{\perp}}& $(1,1)$ & $0$\\
         \hline
    \end{tabular}
    \captionof{table}{Winning probability of the game $\mathcal{G}$}
    \label{table:1}
\end{center}
\end{minipage}
\end{widetext}

One can easily show (Table~\ref{table:1}) that in each case (for any input $(x_0,x_1,y)$), the winning probability is\\ $2 \times \left(\frac{1}{4} \left( 1 + \frac{1}{\sqrt{2}}\right)\right) = \frac{1}{2} \left( 1 + \frac{1}{\sqrt{2}}\right) \approx 0.85$.\\
~\\
Here the winning cases can be divided into two categories. 
\begin{itemize}
    \item For $x_0 = a$, we get $x_y =b$.
    \item For $x_0 = a'$, we get $x_y =b'$.\\
\end{itemize}

As winning probabilities of both the cases are the same, so the game $\mathcal{G}$ with any one of the conditions is sufficient. Here we consider the first case only, i.e., $x_0 = a$, and $x_y =b$. This is because, in the paper \cite{tavakoli}, whenever $x=00$, the state is always $\ket{+}$. However, in game $\mathcal{G}$, whenever $x=00$ the prepared state at Bob's end (before Bob's measurement) can be $\ket{+}$ or $\ket{-}$. Depending on the Alice's measurement, i.e., depending on the value of $a$, the state will collapse to either $\ket{+}$ ($x_0=a$) or $\ket{-}$ ($x_0=a'$). The same will happen for $x=01$, $x=10$ and $x=11$. In the game $\mathcal{G}$, whenever $x_0 = a$ the prepared state at Bob's end matches with the prepared state of the protocol by Tavakoli et al., and we get $b=x_y$. That is why we consider the first case only.\\

\subsection{The Equivalence Proof}
\label{g}
We show the equivalence of CHSH1 and the self-testing presented by Tavakoli et al.~\cite{tavakoli} in two steps. First, we show an equivalence between $\mathcal{G}$ and the protocol of~\cite{tavakoli}. Next, we prove that $\mathcal{G}$ is equivalent to CHSH1.

\begin{theorem}
\label{theo}
The game $\mathcal{G}$ with the winning condition $x_0 = a$ and $x_y =b$ is the entangled version of the self-testing proposed by Tavakoli et al.~\cite{tavakoli}.
\end{theorem}\label{thp1}
\begin{proof}
According to the game $\mathcal{G}$, if $x_0=x_1$, then Alice measures her subsystem in $\{\ket{+},\ket{-}\}$ basis meaning either she gets $\ket{+}$ or $\ket{-}$ with probability $1/2$. Then Bob's state collapses either to $\ket{+}$ or to $\ket{-}$ depending on the measurement result of Alice. 

Similarly, when $x_0\neq x_1$, Alice measures her subsystem in $\{\ket{0},\ket{1}\}$ basis collapsing Bob's state either to $\ket{0}$ or to $\ket{1}$. 
In table \ref{proof}, we show how the quantum state is collapsing at Bob's place in the winning cases.\\ 

\begin{center}
    \begin{tabular}{||c||c|c||c||}
    \hline
        $~x_0 x_1~$ & $~x_0 \oplus x_1~$ & $~a (= x_0)~$ & ~Prepared State\\
        &&&~at Bob's end~ \\
        \hline
        $00$ & $0$ & $0$ & \ket{+}\\
        \hline
        $01$ & $1$ & $0$ & \ket{0}\\
        \hline
        $10$ & $1$ & $1$ & \ket{1}\\
        \hline
        $11$ & $0$ & $1$ & \ket{-}\\
        \hline
    \end{tabular}
  \captionof{table}{Bob's state before measurement in $\mathcal{G}$}
\label{proof}
\end{center}

From the table \ref{proof}, it is clear that the first phase of the game $\mathcal{G}$ boils down to the state preparation phase of the self-testing protocol, i.e.,
\begin{itemize}
    \item[] If $x_0x_1=00$, the prepared state is $\ket{+}$.
    \item[] If $x_0x_1=01$, the prepared state is $\ket{0}$.
    \item[] If $x_0x_1=10$, the prepared state is $\ket{1}$.
    \item[] If $x_0x_1=11$, the prepared state is $\ket{-}$.\\
    \end{itemize}
Now, we consider the measurement phase of Bob in the game $\mathcal{G}$. According to the game, 
\begin{itemize}
        \item If $y=0$, Bob measures his part in $\{\ket{\psi}, \ket{\psi^{\perp}}\}$ basis which implies the projection of the prepared state either on $\ket{\psi}$ or on $\ket{\psi^{\perp}}$. 
        \item If $y=1$, he measures his part in $\{\ket{\phi}, \ket{\phi^{\perp}}\}$ basis which implies the projection of the prepared state either on $\ket{\phi}$ or on $\ket{\phi^{\perp}}$.\\
\end{itemize}
This is exactly the same as the measurement phase of~\cite{tavakoli}. To show that we recall the steps below.
\begin{itemize}
        \item If $y$ is $0$, then the protocol chooses $M_0 = \frac{\sigma_x + \sigma_z}{\sqrt{2}}$ which is actually a projection on the eigenspace of the vector $\vec{v}.\vec{\sigma}$, where $\vec{v}=\frac{1}{\sqrt{2}}(1,0,1)$ and $\sigma_1=\sigma_x$, $\sigma_2=\sigma_y$ and $\sigma_3=\sigma_z$~\cite{nelson}. In this case, the eigenvectors are $\{ \ket{\psi}, \ket{\psi^{\perp}}\}$. 
        \item If $y$ is $1$, then it chooses $M_1 = \frac{\sigma_x - \sigma_z}{\sqrt{2}}$. This is a projection on eigenspace of the vector $\vec{v}.\vec{\sigma}$, where $\vec{v}=\frac{1}{\sqrt{2}}(1,0,-1)$. In this case, the eigenvectors are $\{ \ket{\phi}, \ket{\phi^{\perp}}\}$.
    \end{itemize}
Now, in~\cite{tavakoli}, the authors consider only those events where $(b=x_y|x_0,x_1,y)$. There are eight such possible cases. Hence,
\begin{eqnarray*}
\mathcal{A}= \frac{1}{8} \sum_{x_0,x_1,y} Pr[b=x_y | x_0,x_1,y].
\end{eqnarray*}
In case of the game $\mathcal{G}$, we consider $\Pr[(x_0\oplus x_1)\wedge y=a\oplus b]$ for $x_0 = a$. Our next job is to show that $\Pr[(x_0\oplus x_1)\wedge y=a\oplus b] = \mathcal{A}.$\\
To show this, we consider only the cases where $b=x_y$ in the game $\mathcal{G}$.
\begin{itemize}
\item {\bf \em {Case 1:}} When $y=0$:\\
In this case, $(x_0\oplus x_1)\wedge y=0$ always, whatever the value of $x_i, i\in \{0,1\}$. So to make the game successful, $a$ must be equal to $b$. Hence, we consider the cases, where $a=b=x_0$. Now we consider the probability in two sub-cases, whether $x_0 = x_1$ or $x_0 \neq x_1$. From Table\ref{table:1}, we get $\Pr[(x_0\oplus x_1)\wedge y=a\oplus b|x_0 = x_1, y=0]= [2\times \frac{1}{4}(1+\frac{1}{\sqrt{2}})]= \frac{1}{2} (1+\frac{1}{\sqrt{2}})$ and $\Pr[(x_0\oplus x_1)\wedge y=a\oplus b|x_0 \neq x_1, y=0]= [2\times \frac{1}{4}(1+\frac{1}{\sqrt{2}})]= \frac{1}{2} (1+\frac{1}{\sqrt{2}})$.
\item {\bf \em Case 2:} When $y=1$:\\
In this case, there are two sub-cases.
\begin{itemize}
    \item When $x_0\oplus x_1=0$:\\
    If $x_0\oplus x_1=0$, i.e., when $x_0=x_1$, then $(x_0\oplus x_1)\wedge y=0$. Hence, to make the game successful, $a$ must be equal to $b$. In Table~\ref{table:1}, we consider only the cases where $a=b=x_1$ and get the probability $\Pr[(x_0\oplus x_1)\wedge y=a\oplus b|x_0=x_1,y=1]=2\times \frac{1}{4}(1+\frac{1}{\sqrt{2}})=\frac{1}{2}(1+\frac{1}{\sqrt{2}})$.
    \item When $x_0\oplus x_1=1$:\\
    If $x_0\oplus x_1=1$, i.e., when $x_0\neq x_1$, then $(x_0\oplus x_1)\wedge y=1$. So to make the game successful, $a$ must be the complement of $b$. In this case, there are two possibilities
    \begin{itemize}
        \item $b=x_1=0$, and $a=x_0=1$. In this case, $\Pr[(x_0\oplus x_1)\wedge y=a\oplus b|x_0=1,x_1=0,y=1]= \frac{1}{4}(1+\frac{1}{\sqrt{2}})$.
        \item $b=x_1=1$, then $a=x_0=0$. In this case, $\Pr[(x_0\oplus x_1)\wedge y=a\oplus b|x_0=0,x_1=1,y=1]= \frac{1}{4}(1+\frac{1}{\sqrt{2}})$.
    \end{itemize}
\end{itemize}
\end{itemize}

Hence, the total probability $\Pr[(x_0\oplus x_1)\wedge y=a\oplus b$ is as follows:
\begin{eqnarray*}
&&Pr[(x_0\oplus x_1)\wedge y=a\oplus b)\\
&=& P(x_0 = x_1,y=0)\cdot\frac{1}{2} (1+\frac{1}{\sqrt{2}})\\
&+& P(x_0 \neq x_1,y=0)\cdot\frac{1}{2} (1+\frac{1}{\sqrt{2}})\\
&+& P(x_0=x_1,y=1)\cdot\frac{1}{2}(1+\frac{1}{\sqrt{2}})\\ 
&+& P(x_0 \neq x_1,y=1)[\frac{1}{4}(1+\frac{1}{\sqrt{2}}) +\frac{1}{4}(1+\frac{1}{\sqrt{2}})]\\
&=& \frac{1}{4}\cdot\frac{1}{2}(1+\frac{1}{\sqrt{2}}) +\frac{1}{4}\cdot\frac{1}{2}(1+\frac{1}{\sqrt{2}}) +\frac{1}{4}\cdot\frac{1}{2}(1+\frac{1}{\sqrt{2}})\\
&+& \frac{1}{4}[\frac{1}{4}(1+\frac{1}{\sqrt{2}}) +\frac{1}{4}(1+\frac{1}{\sqrt{2}})]\\
&=&\frac{1}{2}(1+\frac{1}{\sqrt{2}})\\
&\approx& 0.85
\end{eqnarray*}

Hence, we can conclude that $\Pr[(x_0\oplus x_1)\wedge y=a\oplus b] = \mathcal{A}$.
\end{proof}

\begin{theorem}
\label{theorem}
 Game $\mathcal{G}$ is equivalent to CHSH1.
\end{theorem}
\begin{proof}
Now, we show that CHSH1 and $\mathcal{G}$ are equivalent to each other. The only difference between $\mathcal{G}$ and CHSH1 (resp. CHSH) game is that $x$ is a single bit input in CHSH1 whereas in $\mathcal{G}$, $x=x_0x_1$, a two-bit input. However, given $x$, we can always select a random bit $x_0$ and can set $x_1=x\oplus x_0$ where $x, x_0, x_1\in\{0,1\}$. On the other hand, given the two bits $x_0x_1$ of $\mathcal{G}$, one can always reduce them to a single bit $x=x_0\oplus x_1$ of CHSH1. All other conditions of $\mathcal{G}$ are exactly the same as CHSH1. 
\end{proof}

We can conclude the following from Theorem~\ref{theo} and~\ref{theorem}.
 \begin{corollary}
 \label{cor}
The self-testing game (say $\mathcal{T}$) by Tavakoli et al., the game $\mathcal{G}$ and the CHSH1 game are all equivalent.
 \end{corollary}

\section{Our Device Independent QRNG protocol $\mathcal{P}$ deduced from the protocol by Tavakoli et al. \cite{tavakoli}}
\label{proto}

Exploiting the self-testing protocol given by Tavakoli et al., we design a device independent QRNG protocol. In this section, first we discuss the intuitive idea of the protocol and then we present it formally. We have already discussed the self-testing protocol by Tavakoli et al. in preparation and measure scenario. Now we add an extra measurement in $\{\ket{+},\ket{-}\}$ basis in the motivation towards generating a random bit string. This is because whenever we measure $\ket{0}$ or $\ket{1}$ in $\{\ket{+},\ket{-}\}$ basis, the probability of getting $\ket{+}$ or $\ket{-}$ is exactly $\frac{1}{2}$. This is the main tweak of the proposed protocol. So, for the protocol, we take $y \in \{0,1,2\}$, and whenever we get $y=2$, we measure the qubit in $\{\ket{+},\ket{-}\}$ basis. After the measurement, if we get $\ket{+}$ then we  return the classical bit $b=0$ and if we get $\ket{-}$ then we return $b=1$. Hence, whenever we get $x=x_0x_1 \in \{01,10\}$ and $y=2$ we can produce a random bit. For $y\in \{0,1\}$, the corresponding $b$ bits are exploited for self-testing. However, for $x=x_0x_1 \in \{00,11\}$ and $y=2$, a deterministic measurement will happen as $\ket{+}$ or $\ket{-}$ is measured in $\{\ket{+},\ket{-}\}$ basis. Though the produced $b$ bit does not have any randomness in this case, however, those will be used to prove device independence security.

In this protocol we take three different storage spaces (we can think it as an array) to store the three different types of $b$ bits (along with the corresponding tuple $(x_0,x_1,y)$) as follows:\\
\begin{itemize}
\item After storing $b$ and its corresponding $x_0,x_1,y$, we further create three storages to separate the three different events. 
\item $Check$: Whenever $y\in \{0,1\}$ (does not depends on $x$ i.e., $x$ can be anything from the set $\{00,01,10,11\}$), we store the corresponding $b$ bits along with the tuple $(x_0,x_1,y)$ in this storage space.
\item $Rand$: Whenever $x \in \{01,10\}$ and $y=2$, we store the corresponding $b$ bits in this storage space.
\item $False$: Whenever $x \in \{00,11\}$ and $y=2$, we put the corresponding $b$ bits in this storage space.
\item This part is considered as classical post processing part.
\end{itemize}
~\\
The protocol is based on some basic assumptions. Those are enumerated below.\\
\begin{enumerate}
\item The devices make no use of any prior information about the choice of settings $x$ and $y$. (memory loophole)
\item Internal states of the devices are independent and identically distributed (i.i.d).
\item The preparation and measurement devices are independent. (locality loophole)\\
\end{enumerate}
Note that the assumptions are same as the assumptions considered in~\cite{brunner14}. However, there is another assumption that the devices are provided by the trusted vendors. In our case, we remove that assumption by proving that violation of Bell's inequality, randomness extraction and deterministic measurement are not possible simultaneously for any classical correlation in no signalling setting.\\

We are now going to describe our protocol $\mathcal{P}$. It is devided in two parts, quantum part and classical post-processing part.\\
~\\
{\bf Quantum Part of $\mathcal{P}$:}
\begin{enumerate}
\item The algorithm $A$ (device $D_1$) can accept the inputs $x$, where $x=x_0x_1$, $x_i\in\{0,1\}$; $i=\{0,1\}$.

\item The algorithm follows the rule of the state preparation described in~\cite{tavakoli}, i.e., 
  \begin{itemize}
        \item If $x$ is $00$ then the pure state is $\rho_{00} = \frac{\mathds{1} + \sigma_x}{2}$. In other words, $\rho_{00}=\ket{+}\bra{+}$.
        \item If $x$ is $01$ then the pure state is $\rho_{01} = \frac{\mathds{1} + \sigma_z}{2}.$ In other words, $\rho_{01}=\ket{0}\bra{0}$.
        \item If $x$ is $10$ then the pure state is $\rho_{10} = \frac{\mathds{1} - \sigma_z}{2}.$ In other words, $\rho_{10}=\ket{1}\bra{1}$.
        \item If $x$ is $11$ then the pure state is $\rho_{11} = \frac{\mathds{1} - \sigma_x}{2}.$ In other words, $\rho_{11}=\ket{-}\bra{-}$.
    \end{itemize}
    
    \item  The algorithm $B$ (device $D_2$) can accept the input $y$, where $y\in\{0,1,2\}$, and generates an output $b\in\{0,1\}$.
    
\item Depending on $y$, the algorithm $B$ selects the measurement operators as follows.
\begin{itemize}
\item If $y=0$, $B$ chooses $M_0=\frac{1}{\sqrt{2}}(\sigma_x+\sigma_z)$.
\item  If $y=1$, $B$ chooses $M_1=\frac{1}{\sqrt{2}}(\sigma_x-\sigma_z)$.
\item  If $y=2$, $B$ selects $M_2= \sigma_x$.
\end{itemize}

\item $b$ is $0$, if the algorithm measures $\ket{\psi}$ or $\ket{\phi}$ or $\ket{+}$, where $\ket{+}=\frac{1}{\sqrt{2}}(\ket{0}+\ket{1})$ (one of the eigenstates of $\sigma_x$).

\item $b$ is $1$, if the algorithm measures $\ket{\psi^{\perp}}$ or $\ket{\phi^{\perp}}$ or $\ket{-}$, where $\ket{-}=\frac{1}{\sqrt{2}}(\ket{0}-\ket{1})$ (one of the eigenstates of $\sigma_x$).

\item Algorithm $A$ and $B$ repeat $n$ many times, where $n\rightarrow \infty$.
\end{enumerate}
~\\
{\bf Classical Post-Processing part of $\mathcal{P}$:}

\begin{enumerate}
\item After quantum part ends, $A$ exchanges the inputs $x=x_0x_1$ with $B$.

\item The produced $b$ bits (along with the corresponding tuple $(x_0,x_1,y)$) are stored in three different memory locations: $Check$, $Rand$, and $False$ based on the following specifications.
\begin{itemize}
\item For $Check$ bits, $x_0,x_1,y\in \{0,1\}$.
\item For $Rand$ bits, $x_0x_1=\{01,10\}$, $y=2$.
\item For $False$ bits, $x_0x_1=\{00,11\}$, $y=2$.
\end{itemize}

\item $Check$ bits are used to calculate the probability $\mathcal{A}$. If it is $0.85$, the protocol $P$ continues, is aborted otherwise.

\item $Rand$ bits are used as truly random bit sequences.

\item $False$ bits are used to check device independence security.
\end{enumerate}
~\\
\begin{remark}
From the description of our protocol $\mathcal{P}$, one can notice that whenever the input to the algorithm $B$ is $y=2$, we are generating the random numbers. Also one can notice that whenever the input to the algorithm $A$ is $x \in \{00, 11\}$, we cannot generate the random numbers.

If one wishes to buy a device implementing the protocol $\mathcal{P}$, one can run the whole protocol $\mathcal{P}$  several times, until one is fully satisfied that the device is working desirably. After this, they can focus on the random number generation part only. And for only generating the random bit string, the buyer will choose the input $x$ from the set $\{01, 10\}$, and will set $y=2$.

That means, though it seems that out of $12$ possible inputs only $2$ input patterns are useful for randomness generation (making the rate of randomness generation $\frac{1}{6}$), actually once the testing phase is over and the buyer is fully satisfied by the device independence of the protocol, they can modify the input set to make the rate of randomness generation $1$ (modified input set is $x \in \{01, 10\}$ and $y=2$).
\end{remark}

\section{Security Proofs of the QRNG Protocol $\mathcal{P}$}
\label{pprf}
In this section, we prove the device independence security of our protocol $\mathcal{P}$ followed by the proof of randomness generation. In device independence formulation, all quantum devices used in a protocol are considered as black boxes with classical inputs and outputs only. Depending on the input-output statistics, the users should obtain the confidence  about the devices. In this setting, a protocol is called ``secure" if the protocol includes a practical test that guarantees that the users’ quantum devices behave according to the specification, even in the scenario where the devices may have been manufactured by an adversary. That is, we remove the assumption that the devices are supplied by trusted providers. This is the fundamental difference between our protocol $\mathcal{P}$ and the self-testing protocol suggested in~\cite{brunner14}.
\subsection{Proof of Device Independence of our protocol $\mathcal{P}$}
To prove the device independence of our protocol $\mathcal{P}$, we follow the proof technique provided in~\cite{vagiraniacm}. We consider a single randomly chosen round of the protocol. We will prove that if the protocol succeeds, then with high probability the devices’ output in the chosen round must be at least partially unknown to the adversary. In this regard, we will take help of a two player guessing game where one of the participant plays the role of the adversary, i.e., we consider that the respective devices of the player are totally compromised. We will show that the maximum success probability in the guessing game is bounded by the virtue of no signalling principle. Precisely, in the absence of any communication between the players, the output distribution of one player must be independent of the input of the other player. 
To do so, we first show that protocol $\mathcal{P}$ is an augmented CHSH game in Lemma~\ref{lemma1}. In Theorem~\ref{theo1} and~\ref{theo2}, we prove that under no signalling condition the probability to guess the input bit of a player by the other player should not be more that $\frac{1}{2}$. 

\begin{lemma}
\label{lemma1}
Protocol $\mathcal{P}$ is an augmented CHSH game for which the probability of success is  $\frac{2}{3}\omega+\frac{1}{3}$, where $\omega=\cos^2{\pi/8}$. 
\end{lemma}

\begin{proof}
An augmented CHSH game~\cite{vagiraniacm}, is defined as a CHSH game with an extra input $x=2$ at Alice's place and an additional winning condition, i.e., on the input $(x,y)=(2,0)$, Alice and Bob will win the game iff $c=a$, where $c$ is the output at Bob's end on the input $x=2$ at Alice's end.

In our protocol $\mathcal{P}$, we consider the extra input at Bob's side, i.e., $y=2$. As the game is symmetric, we can change Alice to Bob and Bob to Alice. Also in our case $x=x_0x_1$, and the winning condition is $(x_0 \oplus x_1) \cdot y = a \oplus b$ instead of $x \cdot y = a \oplus b$ (check the section \ref{g}). For the simplicity of the proof, we consider $x'=(x_0 \oplus x_1)$. So, from now on whenever we mention $x' = 0$, it means $x \in \{00, 11\}$, and $x' = 1$ means $x \in \{01, 10\}$. Now, we assume that whenever $y=2$, algorithm $B$ will output $c$. We know when $(x',y)=(0,2)$ the base in which the state is prepared by algorithm $A$ and the base which is chosen by algorithm $B$ is same. As the states are measured in their eigen basis, we will get the deterministic result, i.e., if the state is $\ket{+}$ and for that if we assign $a=0$, after measurement $c=0$. Similarly, if the state is $\ket{-}$ and the assigned value for $a=1$, then after measurement $c=1$. Thus, in this case Alice and Bob will always win the game.

Now, we know that for actual CHSH game the success probability is $\omega$ on the inputs $x\in\{0,1\}$ and $y\in\{0,1\}$ and for the augmented CHSH game, on the inputs $(2,0)$, we will always get $c=a$. Thus, the total success probability for the augmented CHSH game is $\frac{2}{3}\omega+\frac{1}{3}$. We already proved that for $x',y \in\{0,1\}$, the protocol $\mathcal{P}$ boils down to the CHSH game and now we showed that for $(x',y)=(0,2)$, it can be reduced to the extra condition of the augmented CHSH game and the success probability is $1$ for this case resulting the total success probability as $\frac{2}{3}\omega+\frac{1}{3}$.
\end{proof}

\begin{theorem}
\label{theo1}
 With reference to Lemma~\ref{lemma1}, algorithm $B$ can not guess the inputs of algorithm $A$ with probability $\frac{1}{2}+\epsilon$, where $\epsilon$ is a non-negligible small fraction.
\end{theorem}

\begin{proof}
Start with Lemma~\ref{lemma1}, let denote the output bit of algorithm $B$ over the input $y=2$ as $c_2$, and $b_y$ over the inputs $y\in\{0,1\}$. Now, as the protocol $\mathcal{P}$ succeeds in augmented CHSH game, when $(x',y)=(0,2)$, $c_2=a_0$, where $a_0$ stands for the output bit supposed to be generated by algorithm $A$ over the input $x'=0$ (readers are suggested to read the reduction proof where it is shown that CHSH game can be reduced to the self-testing protocol given in~\cite{tavakoli} in section \ref{g}). Hence, the relative Hamming distance between $a_0$ and $c_2$ is $d_H(a_0,c_2)=0$. Similarly, according to the protocol for the inputs $(x,y)=(1,2)$, $d_H(a_1,c_2)=1/2$, where $a_1$ is the output bit supposed to be generated by algorithm $A$ over the input $x'=1$. Hence, we can write $d_H(a_0,a_1)=1/2$ on $y=2$. 

Now, in case of CHSH round, i.e., over the inputs $x',y\in\{0,1\}$, $d_H(a_0,b_0)=1-\omega$, $d_H(a_0,b_1)=1-\omega$, $d_H(a_1,b_0)=1-\omega$ and $d_H(a_1,b_1)=\omega$. By the triangle inequality, we get $d_H(a_0,a_1)\gtrsim 2\omega -1$.

Now, we have two conditions; 1) for $x',y\in\{0,1\}$, $d_H(a_0,a_1)\gtrsim 2\omega-1$ and 2) for $x'\in\{0,1\}, y=2$, $d_H(a_0,a_1)=1/2$. Any such correlations cannot be created classically without any prior information of $x'$ to algorithm $B$. In other words, without any communication between the algorithms, more precisely, between the devices which are separated from each other, above two conditions cannot be satisfied simultaneously. This clearly violets the no signalling principle. So algorithm $B$ can not guess the inputs of algorithm $A$ with probability $1$.

Now, we prove that the guessing probability is at most $1/2$. For that, we consider the following trivial guessing game. In the game, given an input $x'\in\{0,1\}$ to one player, other player has to output a bit $b$ (say). The players are allowed to perform any arbitrary computations, but they are not allowed to communicate among themselves. They will win the game if $b=x'$. In this game, the maximum success probability is $1/2$. Any strategy with success probability $\frac{1}{2}+\epsilon$ indicates a violation of no signalling assumption between the players. In our case, if algorithm $B$ can predict the inputs of algorithm $A$ with probability $\frac{1}{2}+\epsilon$, then it clearly violates the no signalling principle. This concludes the proof.
\end{proof}

The immediate corollary of the above lemma is as follows.
\begin{corollary}
If the protocol $\mathcal{P}$ succeeds in augmented CHSH game with probability $\frac{2}{3}\omega+\frac{1}{3}$, where $\omega=cos^2{\pi/8}$, then the algorithm $A$ can not guess the inputs of algorithm $B$ with probability $\frac{1}{2}+\epsilon$.
\end{corollary}

The proof is the same.

\begin{theorem}
\label{theo2}
With reference to Lemma~\ref{lemma1}, algorithm $B$ can guess the outputs of algorithm $A$, i.e., when $\ket{0}$ or $\ket{1}$ or $\ket{+}$ or $\ket{-}$ is emitted over $y=2$ with the probability at most $\frac{3}{4}$.
\end{theorem}

\begin{proof}
Start with Lemma~\ref{lemma1}, let algorithm $B$ follows the strategy that whenever $y=2$, assign $c_2=a_x$. As the protocol $\mathcal{P}$ succeeds in augmented CHSH game, when $(x',y)=(0,2)$, $c_2=a_0$ always. However, for $x'=1$, as the states are measured in an incompatible basis, from the measurement results, i.e., seeing the value of $c_2$, it is impossible to guess the qubit state before measurement. In other words, $\Pr(c_2=a_1)=1/2$. Thus, the overall probability to guess the output of algorithm $A$ by algorithm $B$ over the input $y=2$ is $\frac{1}{2}\cdot 1+\frac{1}{2}\cdot \frac{1}{2}=\frac{3}{4}$. 
\end{proof}

Till now, we have considered  the untrusted devices, i.e., even if there are some trapdoors in the devices, the security remains guaranteed by the virtue of no signalling principle. Now, we consider an external adversary, Eve, who will gather all the extra information through public discussions in the testing phase, i.e., when  $A$ exchanges $x_0x_1$ with $B$ over classical channel, or when $B$ exchanges the value of $b$ to $A$, whenever $x_0x_1 \in \{00, 11\}$ and $y=2$. We will now see if with this extra information, Eve can predict the generated random strings fully or even partially.

Now suppose, from this extra information, Eve has somehow managed to know the particular prepared state $\ket{r}$ ($\ket{r} \in \{\ket{0}, \ket{1}\}$) when $x_0x_1\in\{01,10\}$ and $y=2$. That is she guesses the generation round successfully, i.e., the round when $B$ measures $\ket{r}$ in $\{\ket{+}, \ket{-}\}$ basis. Our claim is that in that case also Eve's guess for $b$ cannot be better than random, i.e., $\Pr(b_{guess}=b)=\frac{1}{2}$. 
This is because, after measuring $\ket{0}$ (similar for $\ket{1}$) in $\{\ket{+}, \ket{-}\}$ basis, probability of getting $\ket{+}$ is $\frac{1}{2}$ and probability of getting $\ket{-}$ is $\frac{1}{2}$. This comes from the fundamental law of Quantum Measurement that is no one can predict the outcome of a measurement a priori with the probability better than $\frac{1}{2}$ (we are focusing on qubit only) even if he/she knows the particular state and the particular basis of the measurement. Hence,$\Pr(b_{guess}=b|x_0,x_1 \in\{01,10\},y=2)=\frac{1}{2}$ conditioning that there is no leakage from Bob's laboratory.

\subsection{Proof of  Randomness of our protocol $\mathcal{P}$}

In our protocol $\mathcal{P}$, presented in section \ref{proto}, the random sequences are extracted from the instances where $x_0x_1\in\{01,10\}$, $y=2$. Precisely, when the states are prepared in $\{\ket{0},\ket{1}\}$ basis and measured in $\{\ket{+},\ket{-}\}$. To show that this extracted bit-sequence is truly random, we use the concept of entropy.


\begin{theorem}
\label{theo5}
The bit-string generated from the protocol $\mathcal{P}$ is truly random.
\end{theorem}
\begin{proof}
In protocol $\mathcal{P}$, whenever $x_0x_1\in\{01,10\}$ and $y=2$ we have prepared a quantum state in $\{\ket{0},\ket{1}\}$ basis and then measure that state in $\{\ket{+},\ket{-}\}$ basis. After measurement, if we get $\ket{+}$ we have returned $0$ and if we get $\ket{-}$ we have returned $1$. Let us consider a random variable $B$ such that $$B = \left\{ \begin{array}{rcl}
0, & \mbox{if the output bit b is 0;}\\
1, & \mbox{if the output bit b is 1.}
\end{array}\right.$$\\

Now,
\begin{eqnarray*}
&&\Pr[B=0]\\
&=& \Pr[\text{After measurement we get \ket{+}}]\\
&=& \Pr[ \text{Before measurement the state was \ket{0}}]\\
&& \cdot \Pr[\text{After measurement we get \ket{+} } |\text{ Before}\\ 
&& \text{measurement the state was \ket{0}}]\\
&& + \Pr[\text{Before measurement the state was \ket{1}}]\\
&& \cdot \Pr[\text{After measurement we get \ket{+} } |\text{ Before}\\
&& \text{measurement the state was \ket{1}}]\\ 
&=& \frac{1}{2} \cdot \frac{1}{2} + \frac{1}{2} \cdot \frac{1}{2}
= \frac{1}{4} + \frac{1}{4}
= \frac{1}{2}
\end{eqnarray*}

Similarly, it can be shown that $\Pr(B=1)=\frac{1}{2}$.

So, we can write
\begin{eqnarray*}
H(B) = \frac{1}{2} \log 2 + \frac{1}{2} \log 2 = \frac{1}{2} + \frac{1}{2} = 1.
\end{eqnarray*}

So, here the uncertainty is maximum (as the entropy is $1$), and hence we can say that the generated bit string is truly random.
\end{proof}

\section{New Self-Testing Protocol Deduced from Pseudo Telepathy Game}
\label{nstg}

In this section we are going to introduce a fully new self-testing protocol  in prepare and measure scenario. At first, we describe our protocol, then we show that our protocol can be deduced from multi-party pseudo-telepathy game.

\subsection{New Self-Testing Protocol in Prepare and Measure Scenario}

We are now going to present our new self-testing protocol. We address the protocol as a new game $G_2$. The description of the game $G_2$ is as follows:\\
~\\
{\bf Description of $G_2$:}

\begin{enumerate}
\item $D_1$ and $D_2$ are two separated devices, and there is a one-way quantum channel from $D_1$ to $D_2$ to pass one quantum state at a time. 

\item There is no classical communication channel between the devices.

\item The devices are placed sufficiently apart so that the no-signalling principal works well.

\item In this game there are two inputs, one is a two-bit input $x_0 x_1$ and another one is a one-bit input $x_2$, where $x_0, x_1, x_2 \in \{0,1\}$. 

\item $x_0x_1$ is an input to the device $D_1$ and $x_2$ is an input to the device $D_2$.

\item After receiving the bits $x_0, x_1,$ the preparation device $D_1$ would prepare some specific quantum state and send it to $D_2$ via the one-way quantum channel. Depending on $x_2$ the device $D_2$ would do some operations on it, and then would measure the quantum state to return a bit $b$.

\item The winning condition of the game $G_2$ is as follows.
    \begin{itemize}
   \item  $ \frac{1}{2}(x_0 + x_1 + x_2) = b + (x_0 \wedge (x_0 \oplus x_1))$, iff the weight of the input string $x_0x_1x_2$ is even.
    \item $ x_1=b$, iff the weight of the input string $x_0x_1x_2$ is odd.
    \end{itemize}
\end{enumerate}
~\\
Let us assume that the event of winning the game $G_2$ is denoted by $W_2$. Then $\Pr[\text{Winning the game } G_2] =\Pr[W_2]$. Also, we define two other events, `$Even$' and `$Odd$' as follows:\\
\begin{itemize}
\item `$Even$': If the weight of the input bit string $x_0x_1x_2$ is even, i.e., $(x_0 + x_1 + x_2) = 0 \text{ mod 2}$.
\item `$Odd$': If the weight of the input bit string $x_0x_1x_2$ is odd, i.e., $(x_0 + x_1 + x_2) = 1 \text{ mod 2}$.\\
\end{itemize}
Then the success probability, $\Pr[W_2 |~ Even]$ of the protocol for the cases where the weight of the input string is even ($\sum_{i=0}^{2} x_i = 0 ~\text{(mod 2)}$) is
\begin{eqnarray*}
\frac{1}{4} \sum_{x_0,x_1,x_2} \Pr\left[\frac{(x_0 + x_1 + x_2)}{2}
= 
b + (x_0 \wedge (x_0 \oplus x_1))\right].\\
\end{eqnarray*}
To achieve the maximum value of $\Pr[W_2 | ~Even]$ we use the following quantum strategy.\\
~\\
{\bf Quantum Strategy for $G_2$:}

\begin{enumerate}
\item After receiving the bits $x_0, x_1$, $D_1$ prepares a qubit as follows:
    \begin{itemize}
        \item If $x_0x_1$ is $00$ then the prepared state is $\ket{+}=\frac{1}{\sqrt{2}}(\ket{0}+\ket{1})$.
        \item If $x_0x_1$ is $01$ then the prepared state is $\frac{1}{\sqrt{2}}(\ket{0} + i \ket{1})$.
        \item If $x_0x_1$ is $10$ then the prepared state is $\frac{1}{\sqrt{2}}(\ket{0} - i \ket{1})$.
        \item If $x_0x_1$ is $11$ then the prepared state is $\ket{-}=\frac{1}{\sqrt{2}}(\ket{0}-\ket{1})$.
    \end{itemize}
    
\item Then $D_1$ sends the prepared state to $D_2$ via the quantum channel.

\item After receiving the quantum states, $D_2$ does the following:
    \begin{itemize}
        \item If $x_2$ is $0$, then it applies H-gate (Hadamard gate) on the prepared state and then measures it in $\{\ket{0},\ket{1}\}$ basis.
        \item If $x_2$ is $1$, then it first applies S-gate followed by a H-gate (Hadamard gate) on the prepared state, and then measures it in $\{\ket{0},\ket{1}\}$ basis. Here, the S-gate is basically as follows: $S(\ket{0})=\ket{0}$ and $S(\ket{1})=i\ket{1}$.
    \end{itemize}
    
    \item The returned bit $b$ is $0$, if it gets $\ket{0}$ after the measurement, and the bit $b$ is $1$, if after the measurement it gets $\ket{1}$.
\end{enumerate}
~\\
We store all the output bits $b$, along with their corresponding input bits $x_0,x_1,x_2$. We run the above procedure for $m$-many times (where $m$ is a sufficiently large number), and then we consider the output bits $b$ for which the input bits satisfy the condition $\sum_{i=0}^{2} x_i = 0 ~\text{(mod 2)}$ and calculate the value of $\Pr[W_2 | ~Even]$. We consider the remaining output bits $b$ corresponding to the odd-weight input strings to calculate $\Pr[b=x_1]$. The explanation behind this probability calculation has been discussed elaborately in the later phase of the paper.\\

In Table~\ref{table:2} , we have shown the probability of getting $\frac{1}{2}(x_0 + x_1 + x_2) = b + (x_0 \wedge (x_0 \oplus x_1))$ for even-weight inputs, and then calculated the value of $\Pr[W_2 |~ Even]$. Here, we have considered  the event $E$ where the condition $\frac{(x_0 + x_1 + x_2)}{2} = b + (x_0 \wedge (x_0 \oplus x_1))$ has been satisfied. Again, as the input bits follow the condition $x_0 + x_1 + x_2 = 0 ~\text{(mod 2)}$, we have only $4$ possible values corresponding to the input string $x_0x_1x_2$ which are $000,011,101$, and $110$. Thus, we have that $\Pr[W_2 | ~Even] = \frac{1}{4} (1+1+1+1) = 1$ for the cases where the weight of the input string is even.\\

\begin{center}
    \begin{tabular}{||c|c|c|c|c||}
        \hline
        $~x_0x_1~$  & $~~\rho~~$ & $~x_2~$ & $~b~$ & $~~~~Pr[E]~~~~$\\
        \hline \hline
        $00$ & $\ket{+}$ & $0$ & $0$ & $1$\\
        \hline
        $01$ & $\frac{1}{\sqrt{2}}(\ket{0} + i \ket{1})$ & $1$ & $1$ & $1$\\
        \hline
        $10$ & $\frac{1}{\sqrt{2}}(\ket{0} - i \ket{1})$ & $1$ & $0$ & $1$\\
        \hline
        $11$ & $\ket{-}$ & $0$ & $1$ & $1$\\
        \hline
    \end{tabular}
    
    \captionof{table}{Calculating the first winning condition for our game $G_2$}
    \label{table:2}
\end{center}

For the inputs satisfying the condition $x_0 + x_1 + x_2 = 1 ~\text{(mod 2)}$ (i.e., the odd weight inputs), we have to calculate the probability $\Pr[b=x_1]$. From Table~\ref{table:2new}, we can see that $\Pr[b=x_1] = \frac{1}{4}(0.5+0.5+0.5+0.5) = 0.5$ for the cases where the input weight is odd.\\

\begin{center}
    \begin{tabular}{||c|c|c|c|c||}
        \hline
        $~x_0x_1~$  & $~~\rho~~$ & $~x_2~$ & $~~Pr[b=x_1]~~$\\
        \hline \hline
        $00$ & $\ket{+}$ & $1$ & $0.5$\\
        \hline
        $01$ & $\frac{1}{\sqrt{2}}(\ket{0} + i \ket{1})$ & $0$ & $0.5$\\
        \hline
        $10$ & $\frac{1}{\sqrt{2}}(\ket{0} - i \ket{1})$ & $0$ & $0.5$\\
        \hline
        $11$ & $\ket{-}$ & $1$ & $0.5$\\
        \hline
    \end{tabular}
    
    \captionof{table}{Calculating second winning condition  for our game $G_2$}
    \label{table:2new}
\end{center}

Hence, we can conclude that using the above quantum strategy, we achieve the success probability $\Pr[W_2 | ~Even]=1$ and $\Pr[b=x_1 | ~Odd]=0.5$. In the next portion of this section we show that our new self-testing protocol can be deduced from the multi-party pseudo-telepathy game. 

\subsection{Deduction From the Multi-Party Pseudo-Telepathy Game}
\label{rmpt}

To show the deduction of the proposed game, $G_2$, we consider the special case of the multi-party pseudo-telepathy game with $n=3$ ($G_1$). The second condition of our self-testing protocol can also been deduced from the multi-party pseudo-telepathy game. In this section, we explain why we add this extra checking condition in our new self-testing protocol.

\subsubsection*{Reduction of the Game $G_2$ (with even weight inputs) from the Multi-Party Pseudo-Telepathy Game $G_1$}

For the reduction proof we first re-visit the quantum strategy of winning the multi-party pseudo-telepathy game $G_1$. In the following discussion we have considered $n=3$, as to show the reduction proof we need $n=3$ only. So we have total 3 players $A_0, A_1,$ and $A_2$. The quantum winning strategy is as follows:\\
~\\
{\bf Quantum Strategy for $G_1$}

\begin{enumerate}
\item Before starting the game the players share an entangled state $\frac{\ket{000} + \ket{111}}{\sqrt{2}}$ among them and $i$-th bit of the entanglement is basically the share of the player $A_i$.

\item After getting the input $x_i$ from the dealer, the player $A_i$ applies the gate S on its share bit iff $x_i = 1$.

\item  Then $A_i$ applies H gate on its share bit and then measure its share bit in $\{\ket{0},\ket{1}\}$ basis.

\item The player $A_i$ outputs $y_i=0$ if it gets $\ket{0}$ after the measurement, and it outputs $y_i=1$ if gets $\ket{1}$ after the measurement.

\item The players win the game if and only if $\sum_{i=0}^{2}y_i = \frac{1}{2}\sum_{i=0}^{2} x_i ~\text{(mod 2)}$.
\end{enumerate} 
~\\
Now, let us look into Table~\ref{table:3} of the multi-party pseudo-telepathy game $G_1$ to get a deeper overview. In the table, the input values satisfy the condition \[\sum_{i=0}^{2} x_i = 0 \text{ (mod 2).}\] 

So the all possible input values come from the set $\{000,011,101,110\}$. Here, in the table, we are basically observing which state is generating at the end of $A_2$ because of the operations and measurements done by both $A_0$ and $A_1$. $A_2$ applies the prescribed operations based on the value of $x_2$ on these states and measure to get the value of $y_2$. From Table~\ref{table:3}, it is clear that because of the states generated at $A_2$'s end, the players always win the game. 
That is, in each case (for any input $(x_0,x_1,x_2)$ with $(x_0+x_1+x_2) = 0 \text{ (mod 2)}$), the winning probability is $4 \times (\frac{1}{4} \times 1) = 1$\\
\begin{widetext}
\begin{minipage}{\linewidth}
\begin{small}
\begin{center}
    \begin{tabular}{||c|c|c|c|c|c|c||}
        \hline
        $~x_0x_1~$ & $~y_0y_1~$ & State at $A_2$'s end & $~x_2~$ & $A_2$'s state after measurement & $y_2$ & $\Pr[\sum_{i=0}^{2}y_i = \frac{1}{2}\sum_{i=0}^{2} x_i ~\text{(mod 2)}]$\\
        \hline\hline
        \multirow{4}{*}{00} & 00 & $\ket{+}$ & \multirow{4}{*}{0} & $\ket{0}$ & 0 & 1\\
        \cline{2-3}\cline{5-6}
        & 01 & $\ket{-}$ & & $\ket{1}$ & 1 & 1\\
        \cline{2-3}\cline{5-6}
        & 10 & $\ket{-}$ & & $\ket{1}$ & 1 & 1\\
        \cline{2-3}\cline{5-6}
        & 11 & $\ket{+}$ & & $\ket{0}$ & 0 & 1\\
         \hline
         \multirow{4}{*}{01} & 00 & $\frac{1}{\sqrt{2}}(\ket{0} + i \ket{1})$ & \multirow{4}{*}{1} & $\ket{1}$ & 1 & 1\\
        \cline{2-3}\cline{5-6}
        & 01 & $\frac{1}{\sqrt{2}}(\ket{0} - i \ket{1})$ & & $\ket{0}$ & 0 & 1\\
        \cline{2-3}\cline{5-6}
        & 10 & $\frac{1}{\sqrt{2}}(\ket{0} - i \ket{1})$ & & $\ket{0}$ & 0 & 1\\
        \cline{2-3}\cline{5-6}
        & 11 & $\frac{1}{\sqrt{2}}(\ket{0} + i \ket{1})$ & & $\ket{1}$ & 1 & 1\\
         \hline
         \multirow{4}{*}{10} & 00 & $\frac{1}{\sqrt{2}}(\ket{0} + i \ket{1})$ & \multirow{4}{*}{1} & $\ket{1}$ & 1 & 1\\
        \cline{2-3}\cline{5-6}
        & 01 & $\frac{1}{\sqrt{2}}(\ket{0} - i \ket{1})$ & & $\ket{0}$ & 0 & 1\\
        \cline{2-3}\cline{5-6}
        & 10 & $\frac{1}{\sqrt{2}}(\ket{0} - i \ket{1})$ & & $\ket{0}$ & 0 & 1\\
        \cline{2-3}\cline{5-6}
        & 11 & $\frac{1}{\sqrt{2}}(\ket{0} + i \ket{1})$ & & $\ket{1}$ & 1 & 1\\
         \hline
         \multirow{4}{*}{11} & 00 & $\ket{-}$ & \multirow{4}{*}{0} & $\ket{1}$ & 1 & 1\\
        \cline{2-3}\cline{5-6}
        & 01 & $\ket{+}$ & & $\ket{0}$ & 0 & 1\\
        \cline{2-3}\cline{5-6}
        & 10 & $\ket{+}$ & & $\ket{0}$ & 0 & 1\\
        \cline{2-3}\cline{5-6}
        & 11 & $\ket{-}$ & & $\ket{1}$ & 1 & 1\\
         \hline
    \end{tabular}
    \captionof{table}{Winning probability of the multi-party pseudo-telepathy game $G_1$}
    \label{table:3}
\end{center}
\end{small}
\end{minipage}
\end{widetext}

From Table~\ref{table:3}, we can observe that the winning cases can be divided into four categories. \\
\begin{itemize}
    \item $y_0=0$ with $y_1=x_0 \wedge (x_0 \oplus x_1)$.
    \item $y_0=0$ with $y_1=1 \oplus (x_0 \wedge (x_0 \oplus x_1))$.
    \item $y_0=1$ with $y_1=x_0 \wedge (x_0 \oplus x_1)$.
    \item $y_0=1$ with $y_1=1 \oplus (x_0 \wedge (x_0 \oplus x_1))$.\\
\end{itemize}

As the winning probability of all the cases here are same, so we can consider any one of those conditions. Here, we consider the first case only, i.e., $y_0=0$ with $y_1=x_0 \wedge (x_0 \oplus x_1)$.

 In our new self-testing protocol, we have prepared the quantum states as follows:\\
\begin{itemize}
\item Quantum state $\ket{+}$ for $x_0x_1=00$.
\item Quantum state $\frac{1}{\sqrt{2}}(\ket{0} + i \ket{1})$ for $x_0x_1=01$.
\item Quantum state $\frac{1}{\sqrt{2}}(\ket{0} - i \ket{1})$ for $x_0x_1=10$.
\item Quantum state $\ket{-}$ for $x_0x_1=11$.\\
\end{itemize}

All the quantum states prepared in our new self-testing protocol $G_2$ is basically the states at $A_2$'s end whenever the conditions $y_0=0$ and $y_1=x_0 \wedge (x_0 \oplus x_1)$ are being satisfied in the multi-party pseudo-telepathy game $G_1$. 

Let us consider that the event of winning the multi-party pseudo-telepathy game $G_1$ is $W_1$. Also we know that the weight of the input string in $G_1$ is always even. So, We can say that $\Pr[\text{Winning the game } G_1] = Pr[W_1 | ~Even]$.

\begin{theorem}
\label{theo6}
The multi-party pseudo-telepathy game $G_1$ is the entangled version of our new game $G_2$ for the even weight inputs.
\end{theorem}\label{thq1}
\begin{proof}
We first show that the state preparation part of our game $G_2$ is basically the states generated at $A_2$'s end in $G_1$.
To show that, we first look into the following table for $G_1$. As we mentioned earlier, that in this case, we are considering the condition where $y_0=0$ and $y_1=x_0 \wedge (x_0 \oplus x_1)$.\\

\begin{center}
    \begin{tabular}{||c||c|c||c||}
    \hline
        $~x_0 x_1~$ & $~x_0 \wedge (x_0 \oplus x_1)~$ & $~y_0y_1~$ & ~Prepared State\\
        &&&~at $A_2$'s end~ \\
        \hline\hline
        $00$ & $0$ & $00$ & $\ket{+}$\\
        \hline
        $01$ & $0$ & $00$ & $\frac{1}{\sqrt{2}}(\ket{0} + i \ket{1})$\\
        \hline
        $10$ & $1$ & $01$ & $\frac{1}{\sqrt{2}}(\ket{0} - i \ket{1})$\\
        \hline
        $11$ & $0$ & $00$ & $\ket{-}$\\
        \hline
    \end{tabular}
    
    \captionof{table}{State at $A_2$'s end before any operation in $G_1$}\label{imdt}
\end{center}

From the above table (Table~\ref{imdt}) we can see that the states prepared at $A_2$'s end before any operation done by $A_2$ is basically same as the prepared states in our new game $G_2$.\\

Now for the measurement part, one can notice that the operations done by $A_2$ is exactly the same as the operations done on the prepared states in the game $G_2$, depending on the value of $x_2$. And also for the measurement and generating the output bit ($y_2$ in game $G_1$ and $b$ in game $G_2$), the same procedure has been followed in the game $G_1$ and in the game $G_2$.\\

Now to conclude the proof, we have to show that $\Pr[winning ~ G_1] = Pr[winning ~ G_2 | (x_0 + x_1 + x_2) = 0 \text{ mod 2}]$, i.e., $\Pr[W_1 | ~Even] = Pr[W_2 | ~Even]$. Actually, from Table~\ref{table:3}, one can easily observe that winning probability in each case is always $1$. So, one can easily calculate that $\Pr[W_1 | ~Even] =1$. Also we have shown that the winning probability in the game $G_2$ with the given quantum strategy is $1$ for even weight inputs (i.e., $\Pr[W_2 | ~Even] = 1$). So, $\Pr[W_1 | ~Even] = \Pr[W_2 | ~Even]$. Numerically, we can write

\begin{eqnarray*}
&& \Pr[W_1 | ~Even]\\
&=& \frac{1}{4} Pr[winning ~ G_1| x_0x_1x_2=000]\\
&& + \frac{1}{4} Pr[winning ~ G_1| x_0x_1x_2=011]\\
&& + \frac{1}{4} Pr[winning ~ G_1| x_0x_1x_2=101]\\
&& + \frac{1}{4} Pr[winning ~ G_1| x_0x_1x_2=110]\\
&=& \frac{1}{4} Pr[\sum_{i=1}^{2}y_i = \frac{1}{2}\sum_{i=1}^{2} x_i ~\text{(mod 2)}| x_0x_1x_2=000]\\
&& + \frac{1}{4} Pr[\sum_{i=1}^{2}y_i = \frac{1}{2}\sum_{i=1}^{2} x_i ~\text{(mod 2)}| x_0x_1x_2=011]\\
&& + \frac{1}{4} Pr[\sum_{i=1}^{2}y_i = \frac{1}{2}\sum_{i=1}^{2} x_i ~\text{(mod 2)}| x_0x_1x_2=101]\\
&& + \frac{1}{4} Pr[\sum_{i=1}^{2}y_i = \frac{1}{2}\sum_{i=1}^{2} x_i ~\text{(mod 2)}| x_0x_1x_2=110].
\end{eqnarray*}

Now in $G_1$, we have two additional winning conditions $y_0=0$ and $y_1=x_0 \wedge (x_0 \oplus x_1)$. So, $\Pr[W_1 | ~Even]$ can be written as follows:

\begin{eqnarray*}
&& \Pr[W_1 | ~Even]\\
&=& \frac{1}{4} Pr[(x_0 \wedge (x_0 \oplus x_1)) + y_2 = \frac{(x_0 + x_1 + x_2)}{2}|\\
&& x_0x_1x_2=000] + \frac{1}{4} Pr[(x_0 \wedge (x_0 \oplus x_1)) + y_2 =\\
&& \frac{(x_0 + x_1 + x_2)}{2}| x_0x_1x_2=011] + \frac{1}{4} Pr[(x_0 \wedge (x_0 \oplus\\
&& x_1)) + y_2 = \frac{(x_0 + x_1 + x_2)}{2}| x_0x_1x_2=101] + \frac{1}{4} Pr[\\
&& (x_0 \wedge (x_0 \oplus x_1)) + y_2 = \frac{(x_0 + x_1 + x_2)}{2}|\\
&& x_0x_1x_2=110].
\end{eqnarray*}

As the prepared state at $A_2$'s end in $G_1$ is same as the prepared state in self-testing game $G_2$, and all the operations done by $A_2$ in $G_1$ is same as the operations done on the prepared state in the game $G_2$, then we can say that $A_2$'s output $y_2$ in $G_1$ is same as the output $b$ in the game $G_2$ (We can also observe that the values of $b$ in Table~\ref{table:2} are matching with corresponding values of $y_2$ in Table~\ref{table:3}). So, we can rewrite the $\Pr[W_1 | ~Even]$ as follows:

\begin{eqnarray*}
&& \Pr[W_1 | ~Even]\\
&=& \frac{1}{4} \cdot (Pr[(x_0 \wedge (x_0 \oplus x_1)) + b = \frac{(x_0 + x_1 + x_2)}{2}|\\ 
&&  x_0x_1x_2=000] + Pr[(x_0 \wedge (x_0 \oplus x_1)) + b = \\
&& \frac{(x_0 + x_1 + x_2)}{2}| x_0x_1x_2=011] + Pr[(x_0 \wedge (x_0 \oplus x_1))\\
&&  + b =\frac{(x_0 + x_1 + x_2)}{2}| x_0x_1x_2=101] + Pr[(x_0 \wedge \\
&& (x_0 \oplus x_1))+ b = \frac{(x_0 + x_1 + x_2)}{2}| x_0x_1x_2=110]\\
&=& Pr[W_2 | ~Even].
\end{eqnarray*}

This concludes the proof of the theorem.
\end{proof}

Now we show that the additional checking condition of our new self-testing protocol has also been deduced from the multi-party pseudo-telepathy game with $n=3$.\\

\subsubsection*{Deduction of The Additional Checking Condition (with odd weight inputs) from Multi-Party Pseudo-Telepathy Game}
\label{mptc}
We know that the input to the multi-party pseudo-telepathy game always has to be of even weight. However, if we consider, the odd-weight inputs and calculate the probability of getting $y_2=x_1$ (this $y_2$ in multi-party pseudo-telepathy game $G_1$ is basically same as $b$ in our new self-testing protocol $G_2$), then we get the value as $0.5$. So if we consider odd weight inputs in $G_1$ and assume the above checking condition as the winning condition for all those cases (we are considering the odd weight inputs only for the shake of the additional checking condition in $G_2$, in reality the weight of inputs in $G_1$ is always even), then we can say that $\Pr[W_1 | ~Odd] = \Pr [y_2=x_1]$. Also $\Pr[W_2 | ~Odd] = \Pr[b=x_1]$. 

In Table~\ref{table:3new}, we calculate the winning probability of the game $G_1$ for odd-weight inputs.\\

\begin{widetext}
    \begin{minipage}{\linewidth}
\begin{center}
    \begin{tabular}{||c|c|c|c|c|c||}
        \hline
        $~x_0x_1~$ & $~y_0y_1~$ & State at $A_2$'s end & $~x_2~$ & $A_2$'s state after applying gates & $Pr[y_2 = x_1]$\\
        \hline\hline
        \multirow{4}{*}{00} & 00 & $\ket{+}$ & \multirow{4}{*}{1} & $\frac{1}{2}((1+i)\ket{0}+(1-i)\ket{1})$ & 0.5\\
        \cline{2-3}\cline{5-5}
        & 01 & $\ket{-}$ & & $\frac{1}{2}((1-i)\ket{0}+(1+i)\ket{1})$ & 0.5\\
        \cline{2-3}\cline{5-5}
        & 10 & $\ket{-}$ & & $\frac{1}{2}((1-i)\ket{0}+(1+i)\ket{1})$ & 0.5\\
        \cline{2-3}\cline{5-5}
        & 11 & $\ket{+}$ & & $\frac{1}{2}((1+i)\ket{0}+(1-i)\ket{1})$ & 0.5\\
         \hline
         \multirow{4}{*}{01} & 00 & $\frac{1}{\sqrt{2}}(\ket{0} + i \ket{1})$ & \multirow{4}{*}{0} & $\frac{1}{2}((1+i)\ket{0}+(1-i)\ket{1})$ & 0.5\\
        \cline{2-3}\cline{5-5}
        & 01 & $\frac{1}{\sqrt{2}}(\ket{0} - i \ket{1})$ & & $\frac{1}{2}((1-i)\ket{0}+(1+i)\ket{1})$ & 0.5\\
        \cline{2-3}\cline{5-5}
        & 10 & $\frac{1}{\sqrt{2}}(\ket{0} - i \ket{1})$ & & $\frac{1}{2}((1-i)\ket{0}+(1+i)\ket{1})$ & 0.5\\
        \cline{2-3}\cline{5-5}
        & 11 & $\frac{1}{\sqrt{2}}(\ket{0} + i \ket{1})$ & & $\frac{1}{2}((1+i)\ket{0}+(1-i)\ket{1})$ & 0.5\\
         \hline
         \multirow{4}{*}{10} & 00 & $\frac{1}{\sqrt{2}}(\ket{0} + i \ket{1})$ & \multirow{4}{*}{0} & $\frac{1}{2}((1+i)\ket{0}+(1-i)\ket{1})$ & 0.5\\
        \cline{2-3}\cline{5-5}
        & 01 & $\frac{1}{\sqrt{2}}(\ket{0} - i \ket{1})$ & & $\frac{1}{2}((1-i)\ket{0}+(1+i)\ket{1})$ & 0.5\\
        \cline{2-3}\cline{5-5}
        & 10 & $\frac{1}{\sqrt{2}}(\ket{0} - i \ket{1})$ & & $\frac{1}{2}((1-i)\ket{0}+(1+i)\ket{1})$ & 0.5\\
        \cline{2-3}\cline{5-5}
        & 11 & $\frac{1}{\sqrt{2}}(\ket{0} + i \ket{1})$ & & $\frac{1}{2}((1+i)\ket{0}+(1-i)\ket{1})$ & 0.5\\
         \hline
         \multirow{4}{*}{11} & 00 & $\ket{-}$ & \multirow{4}{*}{1} & $\frac{1}{2}((1-i)\ket{0}+(1+i)\ket{1})$ & 0.5\\
        \cline{2-3}\cline{5-5}
        & 01 & $\ket{+}$ & & $\frac{1}{2}((1+i)\ket{0}+(1-i)\ket{1})$ & 0.5\\
        \cline{2-3}\cline{5-5}
        & 10 & $\ket{+}$ & & $\frac{1}{2}((1+i)\ket{0}+(1-i)\ket{1})$ & 0.5\\
        \cline{2-3}\cline{5-5}
        & 11 & $\ket{-}$ & & $\frac{1}{2}((1-i)\ket{0}+(1+i)\ket{1})$ & 0.5\\
         \hline
    \end{tabular}
    \captionof{table}{Probability calculation for odd weight inputs for $G_1$}
    \label{table:3new}
\end{center}
\end{minipage}
\end{widetext}

From Table~\ref{table:3new}, we can see that $\Pr[W_1 | ~Odd] = \Pr [y_2=x_1] = 0.5$. We already get $\Pr[W_2 | ~Odd] = Pr[b=x_1] = 0.5$ (see Table~\ref{table:2new}). That means $\Pr[W_1 | ~Odd] = \Pr[W_2 | ~Odd]$. Hence, the additional checking condition of our new self-testing protocol can be derived from the three-party pseudo-telepathy game too.

\subsubsection*{Importance of additional checking condition of our new self-testing game}

In the multi-party pseudo-telepathy game, the input bits $x_0$ and $x_1$ have been sent to $A_0$ and $A_1$ respectively and $A_0$ and $A_1$ do not communicate among themselves after getting the inputs. But in our game $G_2$, $D_1$ has the access to both the inputs $x_0$ and $x_1$. And hence there exists a strategy which can mimic the success probability of multi-party pseudo-telepathy game $G_1$. The strategy is as follows:\\
\begin{itemize}
\item For any input $x_0x_1$ to $D_1$, $D_1$ prepares the state $\ket{x_1}$ and send this to $D_2$.
\item $D_2$ does not do any operation on the received qubit and just measure it in $\{\ket{0},\ket{1}\}$ basis to output $b$.\\
\end{itemize}
However, in this case, along with the probability $\Pr[W_2 | ~Even]=1$, $\Pr[b=x_1 | ~Odd]=1$ too. So, the condition for the even weight inputs is not sufficient to check the quantumness of the devices like multi-party pseudo-telepathy game. To prevent this type of strategy, we have to impose the extra condition for odd weight inputs. In the next section, we will discuss the additional condition and show how it can be also derived from the multi-party pseudo-telepathy game.

\ \\
Now, we consider the Pseudo-Telepathy game $G_1$ with this extra condition and name it as {\em Augmented Pseudo-Telepathy} game.

\begin{lemma}
The success probability of augmented pseudo-telepathy game is $\frac{3}{4}$.
\end{lemma}

\begin{proof}
The success probability of pseudo-telepathy game for the even weight inputs is $1$, whereas the success probability that $y_2=x_1$ in the pseudo telepathy game for odd weight inputs is $1/2$ from the virtue of no signalling condition. This extra condition can be reduced to the trivial two party guessing game where one party guesses the inputs of the other party.  We already showed that with no signalling assumption, the success probability of such guessing game is $1/2$. Thus, the over all success probability of the augmented pseudo-telepathy game is $\frac{1}{2}\cdot 1+\frac{1}{2}\cdot \frac{1}{2}=\frac{3}{4}$.
\end{proof}

\begin{lemma}
The self-testing protocol $G_2$ is an augmented pseudo-telepathy game.
\end{lemma}

\begin{proof}
We already showed that pseudo-telepathy game $G_1$ can be reduced to $G_2$. We also proved that $\Pr(W_1|Even)=\Pr(W_2|Even)$ and $\Pr(W_1|Odd)=\Pr(W_2|Odd)$. Let the event of winning in augmented $G_1$ be $AW_1$ and the event of winning in $G_2$ be $AW_2$. Then, it is easy to show that $\Pr(AW_1)=\Pr(AW_2)$. This concludes the proof.
\end{proof}
\section{Our Device Independent QRNG Protocol $\mathcal{Q}$ based on the New Self-Testing Protocol $G_2$}
\label{NQRNG}

In this section, from the new self-testing protocol given in section \ref{nstg}, we have designed a device independent QRNG $\mathcal{Q}$. 

In this protocol, we have used two different storage spaces (we can think it as an array) to store the $b$ bits and the corresponding tuple $(x_0,x_1,x_2)$ and proceed  as follows:\\
\begin{itemize}
\item After storing $b$ and its corresponding $x_0,x_1,x_2$, we further create two storages to separate the two different events. 
\item $Check$: Whenever $(x_0+x_1+x_2) = 0 \text{ (mod 2)}$, we store the corresponding $b$ bits along with the tuple $(x_0,x_1,x_2)$ in this storage space.
\item $Rand$: Whenever $(x_0+x_1+x_2) = 1 \text{ (mod 2)}$, we store the the corresponding $b$ along with the tuple $(x_0,x_1,x_2)$ bits in this storage space.
\item This part is considered as classical post processing part.\\
\end{itemize}

Now, we enumerate the protocol $\mathcal{Q}$, in two parts namely quantum part and classical post-processing part, as follows: \\
~\\
{\bf Quantum Part of $\mathcal{Q}$:}

\begin{enumerate}
\item Algorithm $A$ (device $D_1$) accepts the inputs $x_0,x_1$, and algorithm $B$ (device $D_2$) accepts the input  $x_2$, where $x_i\in\{0,1\}$ for $i\in\{0,1,2\}$.

 \item Algorithm $B$ generates the output $b\in\{0,1\}$.
 
\item The algorithm $A$ follows the rule of the state preparation described in our new self-testing game, i.e., 
  \begin{itemize}
        \item If $x_0x_1$ is $00$ then the prepared state is $\ket{+}=\frac{1}{\sqrt{2}}(\ket{0}+\ket{1})$.
        \item If $x_0x_1$ is $01$ then the prepared state is $\frac{1}{\sqrt{2}}(\ket{0} + i \ket{1})$.
        \item If $x_0x_1$ is $10$ then the prepared state is $\frac{1}{\sqrt{2}}(\ket{0} - i \ket{1})$.
        \item If $x_0x_1$ is $11$ then the prepared state is $\ket{-}=\frac{1}{\sqrt{2}}(\ket{0}-\ket{1})$.
    \end{itemize}
    
\item Depending on $x_2$, the algorithm $B$ does the following:
\begin{itemize}
\item If $x_2=0$, the algorithm $B$ applies only H gate on the prepared qubit.
\item If $x_2=0$, the algorithm $B$ applies S gate followed by a H gate on the prepared qubit.
\end{itemize}

\item Then algorithm $B$ measures the qubit in $\{\ket{0},\ket{1}\}$ basis.

\item The output bit $b$ is $0$, if the algorithm measures $\ket{0}$.

\item The output bit $b$ is $1$, if the algorithm measures $\ket{1}$.

\item Algorithm $A$ and algorithm $B$ are repeated $n$ many times, where $n$ is a sufficiently large number (ideally $n \rightarrow \infty$).
\end{enumerate}
\newpage
{\bf Classical Post-Processing part of $\mathcal{Q}$:}

\begin{enumerate}
\item After quantum part ends, $A$ exchanges the inputs $x_0$, $x_1$ with $B$.

\item The produced $b$ bits (along with the corresponding tuple $(x_0,x_1,x_2)$) are stored in two different storage spaces and further processed depending on the even (odd) weight of the inputs. The two new memory locations; $Check$ and $Rand$ is created based on the following specifications.
\begin{itemize}
\item For $Check$ bits, $(x_0+x_1+x_2) = 0 \text{ (mod 2)}$.
\item For $Rand$ bits, $(x_0+x_1+x_2) = 1 \text{ (mod 2)}$.
\end{itemize}

\item $Check$ bits are used to calculate the probability $\Pr[W_2 | ~Even]$. If it is $1$, the protocol $\mathcal{Q}$ continues, aborts otherwise.

\item A fraction $\gamma\geq \frac{1}{2}$ of the $Rand$ bits are used to check whether $\Pr[b=x_1]=0.5$ or not. If the probability is not $0.5$, then we abort the protocol. Otherwise, we use the remaining $1-\gamma$ $Rand$ bits as truly random bit sequences.\\
\end{enumerate}
~\\
Here also for the proper execution of the protocol we need the similar assumptions as before (the assumptions for the new self-testing protocol).\\
\begin{itemize}
\item There are two separate devices $D_1$ and $D_2$.
\item $D_1$ takes the input $x_0x_1$ and prepares a quantum state $\rho$.
\item Then $D_1$ sends that quantum state $\rho$ to $D_2$.
\item $D_2$ takes the input $x_2$ and the quantum state $\rho$, and performs all the operations and measurements. It outputs the bit $b$.
\item There is only one secure one-way quantum channel from $D_1$ to $D_2$ to pass the quantum state $\rho$. 
\item There does not exist any classical channel between $D_1$ and $D_2$. 
\item The devices make no use of any prior information about the choice of settings $x_i$, where $i=0,1,2$.
\item Internal states of the devices are independent and identically distributed (i.i.d).
\item The preparation (device $D_1$) and measurement device (device $D_2$) are independent.\\
\end{itemize}

\begin{remark}
Here as we need both the $Check$ bits and the $Rand$ bits to test the device independence, so it seems that each time whenever we try to generate an $k$-bit random bit string from the protocol we have to repeat the step-7 for $4k+\delta$ many times (where $\delta$ is a non-zero integer). Note that the buyer just needs to test the device independence of the QRNG for the first time only. That is he/she has to check whether the dealer has fooled him/her with a defective or corrupted QRNG. After the testing is done and the buyer is assured that the QRNG is accurate, he/she only repeats the step-7 for $2k+\delta$ many times and use the $Rand$ bits as a random bit string of length $k$. So, here the actual rate of randomness generation is $\frac{1}{2}$.
\end{remark}
\section{Security Proofs of the QRNG Protocol $\mathcal{Q}$}
\label{prf}
In this section, we will prove the device independence  followed by the randomness extraction of the protocol $\mathcal{Q}$.
\subsection{Proof of Device Independence of our protocol $\mathcal{Q}$} We prove the device independence of protocol $\mathcal{Q}$ with the help of Lemma~\ref{lemma} and Theorem~\ref{theo3}.  

\begin{lemma}
\label{lemma}
The protocol $\mathcal{Q}$ for the odd weight input strings can be reduced to the trivial two players' guessing game where upon an input $x\in\{0,1\}$ given to the player $A$, player $B$ will try to guess the bit . If the guessed bit $b=x$, both the players win the game provided there is no communication between the players after the game starts.
\end{lemma}

\begin{proof}
Protocol $\mathcal{Q}$ consists of two algorithms; algorithm $A$ and algorithm $B$. These two algorithms can be viewed as two players $A$ and $B$ respectively. The only difference between the protocol $\mathcal{Q}$ for the odd weight input strings and the trivial guessing game is that algorithm $A$ can accept two inputs $x_0,x_1$ and algorithm $B$ needs to guess one input, i.e., here, the winning condition is $x_1=b$, where $b$ is the output bit of algorithm $B$. However, we can replace $x_1$ by $x_0$. It can be shown that assuming no signalling, $\Pr(x_0=b)=1/2$ too. This concludes the proof.
\end{proof}

\begin{theorem}
\label{theo3}
Occurrence of the following two conditions simultaneously certifies the device independence security of the protocol $\mathcal{Q}$. 
\begin{enumerate}
\item $\Pr(\frac{1}{2}(x_0 + x_1 + x_2) = b + (x_0 \wedge (x_0 \oplus x_1)))=1$ for even weight strings. (From Pseudo Telepathy Game)
\item $\Pr(x_1=b)=0.5$ for odd weight strings. (From Lemma~\ref{lemma})
\end{enumerate}
\end{theorem}

\begin{proof}
The success probability of the players in the trivial guessing game is bounded to $1/2$ by the virtue of no signalling principle. We proved that the protocol $\mathcal{Q}$ for the odd weight strings can be reduced to this trivial guessing game by exploiting the condition $\Pr(x_1=b)=0.5$. 
Now, for $\Pr(\frac{1}{2}(x_0 + x_1 + x_2) = b + (x_0 \wedge (x_0 \oplus x_1)))=1$, 
where, $x_0,x_1$ are the inputs of algorithm $A$ and $x_2$ and $b$ are the input and output of algorithm $B$ respectively, we consider the following strategy (section~\ref{mptc}). 
\begin{itemize}
\item For any inputs $x_0x_1$ to algorithm $A$ (device $D_1$), prepare the state $\ket{x_1}$ and send this to algorithm $B$ (device $D_2$).
\item Algorithm $B$ does not do any operation on the received qubit, but just measures it in $\{\ket{0},\ket{1}\}$ basis to output $b$. 
\end{itemize}
If the algorithms follow this strategy, we will obtain $\Pr(\frac{1}{2}(x_0 + x_1 + x_2) = b + (x_0 \wedge (x_0 \oplus x_1)))=1$ for even weight strings. Now, as the basis for the state preparation and the basis for the measurement are same here, the relative Hamming distance $d_H(x_1,b)=0$. However, to satisfy the second condition we must require $\Pr(x_1=b)=0.5$, i.e., $d_H(x_1,b)=1/2$ for the odd weight input strings. It is not possible unless the adversary successfully guesses the `check' rounds and the `rand' rounds. In this case, she will play the above strategy for the check rounds or for the even weight strings, and for the 'rand' rounds, i.e., for the odd weight strings, she will output a random bit. This clearly violates no signalling assumption.
Hence, the occurrence of both the following conditions simultaneously certifies the device independence security of the protocol $\mathcal{Q}$ under no signalling assumption.
\begin{enumerate}
\item $\Pr(\frac{1}{2}(x_0 + x_1 + x_2) = b + (x_0 \wedge (x_0 \oplus x_1)))=1$. (Pseudo Telepathy Game )
\item $\Pr(x_1=b)=0.5$. (Lemma~\ref{lemma})
\end{enumerate}
 
\end{proof} 
One should note that nowhere in the proof, we make any assumptions on the fabrication of the devices. Neither do we assume that the devices are supplied by a trusted provider. We treat the algorithms ($A$ as well as $B$) as black boxes and check the required conditions based on the input-output statistics of those algorithms.\\

Similar to our protocol $\mathcal{P}$, we now consider the effect of the external adversary, Eve, in the protocol $\mathcal{Q}$. In this direction, our claim is that even if the adversary gathers extra information from the testing phase and manages to know about the particular prepared state in any round of generation phase and the input of $B$ in that round, then also he/she cannot guess the generated bit $b$ in that round with a probability better than $\frac{1}{2}$. In other words, $\Pr(b_{guess}=b|E)=\frac{1}{2}$, where $E$ is the side information leaked in the testing phase to Eve. This is because, in the generation phase, the states $\frac{1}{2}((1\pm i)\ket{0}+(1\pm i)\ket{1})$ are measured in $\{\ket{0},\ket{1}\}$ basis resulting either $\ket{0}$ or $\ket{1}$ with probability $\frac{1}{2}$.\\

Let $B$ be the random variable corresponding to the output $b$ for the inputs $\{001,010,100,111\}$ and $B_E$ be the random variable corresponding to the inputs $\{000,011,101,110\}$. Iff the laboratory of Bob is fully sealed and the raw bits $b$ are not leaked to Eve, then $I(B,B_E)$, i.e., the mutual information of Eve about the raw bits in the generation round given the side information is bounded by the principle of quantum measurement and in ideal case, it will be zero.

\subsection{Proof of  Randomness of our protocol $\mathcal{Q}$}

In our protocol $\mathcal{Q}$,  the random bits are generated from the instances where the weight of the input bit-string $x_0x_1x_2$ is odd. Now, we prove the true randomness of the extracted bit-sequence by exploiting the idea of entropy.

\begin{theorem}
\label{theo8}
The bit-string generated from the protocol $\mathcal{Q}$ is truly random.
\end{theorem}
\begin{proof}
Let $B$ be the random variable corresponding to the output $b$, when the input comes from the set $\{001,010,100,111\}$. Clearly the values of $b$ can be either $0$ or $1$, so the event space for the random variable $B$ is $E=\{0,1\}$. Then entropy of B is $$H(B) = \sum_{b \in E} Pr[B=b] \cdot log \left(\frac{1}{Pr[B=b]}\right).$$
Before we compute the values of $\Pr[B=0]$ and $\Pr[B=1]$ to compute $H(B)$, we calculate the probability of getting $\ket{0}$ or $\ket{1}$ after measuring the states $\frac{1}{2}((1+i)\ket{0}+(1-i)\ket{1})$ or $\frac{1}{2}((1-i)\ket{0}+(1+i)\ket{1})$ in $\{\ket{0},\ket{1}\}$ basis, in Table~\ref{rand}.

\begin{center}
\begin{scriptsize}
\begin{tabular}{||c|c|c|c||}
\hline
Input & Before & After & Probability\\
& Measurement & Measurement & \\
\hline\hline
\multirow{2}{*}{001 or 010} & \multirow{2}{*}{$\frac{1}{2}((1+i)\ket{0}+(1-i)\ket{1})$} & $\ket{0}$ & 0.5\\
\cline{3-4}
& & $\ket{1}$ & 0.5\\
\hline
\multirow{2}{*}{100 or 111} & \multirow{2}{*}{$\frac{1}{2}((1-i)\ket{0}+(1+i)\ket{1})$} & $\ket{0}$ & 0.5\\
\cline{3-4}
& & $\ket{1}$ & 0.5\\
\hline
\end{tabular}
\captionof{table}{Probability of Measurement in $\{\ket{0},\ket{1}\}$ basis for our protocol $\mathcal{Q}$}
\label{rand}
\end{scriptsize}
\end{center}

Now let us compute the values of $Pr[B=0]$ and $Pr[B=1]$. We get,\\

\begin{eqnarray*}
&&\Pr[B=0]\\
&=& \Pr[b=0 | (x_0x_1x_2)\in\{001,010,100,111\}]\\
&=&\Pr[\text{Measuring }\ket{0} | (x_0x_1x_2) \in\{001,010,100,111\}]\\
&=&\frac{2}{4} \cdot \Pr[\text{Measuring }\ket{0}|(x_0x_1x_2)
\in\{001,010\}]\\
&& + \frac{2}{4} \cdot \Pr[\text{Measuring }\ket{0} | 
(x_0x_1x_2)\in\{100,111\}]\\
&=&\frac{1}{2} \cdot \Pr[\text{Measuring}\ket{0} |\ket{\psi}= \frac{1}{2}((1+i)\ket{0}+(1-i)\\
&& \ket{1})] + \frac{1}{2} \cdot \Pr[\text{Measuring }\ket{0} | \ket{\psi}= \frac{1}{2}((1-i)\ket{0}+\\
&& (1+i)\ket{1})]\\
&=& \frac{1}{2} \cdot \frac{1}{2} + \frac{1}{2} \cdot \frac{1}{2}=\frac{1}{2}.\\
\end{eqnarray*}

Similarly, it can be shown that $\Pr(B=1)=\frac{1}{2}$.
Hence, we can write $$H(B) = \frac{1}{2} \cdot log(2) + \frac{1}{2} \cdot log(2) = 1.$$
As the uncertainty is maximum here,  we can say that the generated bit string is truly random.
\end{proof} 


\section{Comparative Analysis of Our Protocols}
\label{compare}
Now we compare our DI-QRNG protocols in prepare and measure scenario. In our protocol $\mathcal{Q}$, we need more qubits compared to the protocol $\mathcal{P}$. However, the difference between the best possible classical strategy and the best possible quantum strategy is $0.15$ (classical the winning probability is $\frac{1}{2} + 2^{-\lceil \frac{n}{2}\rceil}$) which is more efficient distinguisher than CHSH game or any variants of it.\\

For comparison, we consider the protocols presented in the papers~\cite{rp1,rp2} as those are proposed in P\&M framework and like ours these are state dependent protocols. We name the semi-device independent random number expansion protocol as SQRNG and the experimental measurement-device-independent random number generation protocol as MQRNG. 

In this regard, one should note that the comparison is provided for the raw bit stream generated from the protocols. We do not consider any kind of privacy amplifications on the raw randomness.\\

\begin{widetext}
    \begin{minipage}{\linewidth}
\begin{center}
    \begin{tabular}{|c|c|c|c|c|c|}
        \hline
         & \# Qubits & Rate of & Prepare/& Difference&State\\
        Protocol & needed for & Randomness & Measure & of $Q_{val}$&Dependent/\\
        Name & $k$-bit & Generation & device & and $C_{val}$&Independent\\
         & randomness & & trustability & &\\
        \hline \hline
        & & & Needed & &\\
        SQRNG~\cite{rp1} & $k+\delta$ & 1 & for & $0.1$&State \\
         & & & Measure & &Dependent\\
        \hline
        & & & Needed && \\
        MQRNG~\cite{rp2} & $k+\delta$ & 1 & for & NA&State\\
         & & & Prepare &&Dependent \\
        \hline
        Our & & & & &\\
        Protocol $\mathcal{P}$ & $k+\delta$ & 1 & Not & $0.1$&State\\
         & & & Needed &&Dependent \\
        \hline
        Our & & & && \\
        Protocol $\mathcal{Q}$ & $2k+\delta$ & $\frac{1}{2}$ & Not & $0.15$&State\\
         & & & Needed &&Dependent \\
        \hline
    \end{tabular}
    \captionof{table}{Comparison of our protocols with the existing ones}
    \label{table:4}
\end{center}
\end{minipage}
\end{widetext}

In Table~\ref{table:4}, $Q_{val}$ and $C_{val}$ denote the winning probability using the quantum strategy and the best possible classical strategy respectively and $\delta$ is a fixed non-zero integer. Also the second column of the table 
denotes the number of qubits needed for generating $k$-bit random number, after 
the initial validation of the device independence property has been completed.

\section{conclusion}
\label{conl}
In the current manuscript, we have shown that the self-testing game proposed by Tavakoli et al. (Phys. Rev. A, 2018) is basically a reduced version of the well-known CHSH game. Though they have not used any entanglement in their protocol, however, it indirectly relies on the violation of Bell's inequality. We have exploited their self-testing protocol to design a DI-QRNG protocol and proved its device-independence and true-randomness of the outputs. Further, we have proposed a new and more efficient self-testing game in prepare and measure scenario which is based on the pseudo-telepathy game. Based on this new game, we have presented a novel DI-QRNG protocol and proved its device-independence as well as true-randomness extraction.

For our new self-testing protocol, we have used the multi-party pseudo-telepathy game for 3 players. One can observe that the winning probability of the multi-party pseudo-telepathy game with $n$ players for the classical strategy is $\frac{1}{2} + 2^{-\lceil \frac{n}{2}\rceil}$. Therefore, as $n$ grows larger, the winning probability difference between the classical and the quantum strategies moves closer to $0.5$. So, it will be easier to distinguish whether the classical or the quantum strategy has been exploited. In this direction, the generalization of our protocol for $n$ beyond 3 could be an interesting future research work.


\end{document}